\documentclass[sigconf,edbt]{acmart-edbt2021}

\def\BibTeX{{\rm B\kern-.05em{\sc i\kern-.025em b}\kern-.08em
    T\kern-.1667em\lower.7ex\hbox{E}\kern-.125emX}}

\acmISBN{978-3-89318-084-4}

\acmConference[EDBT 2021]{24th International Conference on Extending Database Technology (EDBT)}{March 23-26, 2021}{Nicosia, Cyprus}
\acmYear{2021}

\acmConference[EDBT]{EDBT}{2021}{Cyprus}

\usepackage{enumitem}

\usepackage{mdframed}
\usepackage{framed}
\usepackage{hyphenat}
\usepackage{algpseudocode}
\usepackage{epsfig}
 \usepackage[boxed,ruled,vlined,linesnumbered,lined, noend]{algorithm2e}
\usepackage{setspace}
\usepackage{enumitem}
\usepackage{ragged2e}
\usepackage{multirow}
\usepackage{hhline}
\usepackage{caption}
\usepackage{subcaption}
\usepackage{graphicx}
\usepackage{wrapfig}

\setcopyright{rightsretained}
\usepackage{tikz}

\usepackage{csquotes}
\usepackage{enumitem}
\usepackage{tablefootnote}
\usepackage{fnpct}
\usepackage{bbm}
\usepackage{lipsum}

\newcommand{\BigO}[0]{{\mathcal{O}}\xspace}

\makeatletter
\def\thm@space@setup{\thm@preskip=0pt
\thm@postskip=0pt}
\makeatother

\usepackage{array}
\usepackage{makecell}

\newcommand{\B}{\vspace*{-\smallskipamount}}
\newcommand{\BB}{\vspace*{-\medskipamount}}
\newcommand{\BBB}{\vspace*{-\bigskipamount}}

\usepackage{longtable}
\usepackage{mdframed}

\usepackage[hang,flushmargin]{footmisc}
\usepackage{natbib}
\newcommand\encircle[1]{%
\tikz[baseline=(X.base)]
   \node (X) [draw, shape=circle, inner sep=-1.5pt, fill=black, text=white] {\strut #1};}

\newcommand*\wrapletters[1]{\wr@pletters#1\@nil}
\def\wr@pletters#1#2\@nil{#1\allowbreak\if&#2&\else\wr@pletters#2\@nil\fi}
\usepackage{tikz}
\usetikzlibrary{automata,matrix,shapes,arrows,positioning,chains,calc}
\usetikzlibrary{snakes}
\usetikzlibrary{arrows,scopes}
\usetikzlibrary{positioning,chains,fit,shapes,calc}
\usepackage{caption}
\usepackage{subcaption}
\usepackage{amsmath}

\begin{document}
\title{\textsc{Concealer}: SGX-based Secure, Volume Hiding, and \\ Verifiable Processing of Spatial Time-Series Datasets }
\author{Peeyush Gupta, Sharad Mehrotra, Shantanu Sharma, Nalini Venkatasubramanian, and \\  Guoxi Wang}
\affiliation{\institution{University of California, Irvine, USA.\thanks{\textbf{A preliminary version of this paper has been accepted in the  24th International Conference on Extending Database Technology (EDBT) 2021.} \\ This material is based on research sponsored by DARPA under Agreement No. FA8750-16-2-0021.
The U.S. Government is authorized to reproduce and distribute reprints for Governmental purposes notwithstanding any copyright notation thereon. The views and conclusions contained herein are those of the authors and should not be interpreted as necessarily representing the official policies or endorsements, either expressed or implied, of DARPA or the U.S. Government.
This work is partially supported by NSF Grants No. 1527536,  1545071, 2032525, 1952247, 1528995, and 2008993. }}  \texttt{sharad@ics.uci.edu, shantanu.sharma@uci.edu}}
\setlength{\belowdisplayskip}{0pt}
\setlength{\belowdisplayshortskip}{0pt}
\setlength{\abovedisplayskip}{0pt}
\setlength{\abovedisplayshortskip}{0pt}

\renewcommand{\shortauthors}{Gupta, Mehrotra, Sharma, Venkatasubramanian, and Wang}

\begin{abstract}
This paper proposes a system, entitled \textsc{Concealer} that allows sharing time-varying spatial data (\textit{e}.\textit{g}., as produced by sensors) in encrypted form to an untrusted third-party service provider to provide location-based applications (involving aggregation queries over selected regions over time windows) to users.
\textsc{Concealer} exploits carefully selected encryption techniques to use indexes supported by database systems and combines ways to add fake tuples in order to realize an efficient system that protects against leakage based on output-size. Thus, the design of \textsc{Concealer} overcomes two limitations of existing symmetric searchable encryption (SSE) techniques:
(\textit{i}) it avoids the need of specialized data structures that limit usability/practicality of SSE in large scale deployments, and
(\textit{ii}) it avoids information leakages based on the output-size, which may leak data distributions.
Experimental results validate the efficiency of the proposed algorithms over a spatial time-series dataset (collected from a smart space) and TPC-H datasets, each of 136 Million rows, the size of which prior approaches have not scaled to.
\end{abstract}
\maketitle

\section{Introduction}
\label{sec:Introduction}

\bgroup
\def\arraystretch{1.2}
\begin{table*}[!t]
\BBB
\scriptsize
\centering
\begin{tabular}{| l | l| l|l | l |l | l |l|}
\hline
\textbf{Techniques} & \textbf{Frequent and fast insertion} & \textbf{Fast query execution} & \textbf{DBMS supported index} & \multicolumn{3}{|c|}{\textbf{Preventing attacks}} \\ \hline

    &  &   & & \textbf{Data distribution}    & \textbf{Output-size}   & \textbf{Access-patterns}  \\\hline

DET (Always Encrypt~\cite{DBLP:conf/sigmod/AntonopoulosASE20}) & 1 & 1 & \textcolor[rgb]{0.00,0.50,0.00}{Yes} & \textcolor[rgb]{1.00,0.00,0.00}{No} & \textcolor[rgb]{1.00,0.00,0.00}{No} & \textcolor[rgb]{1.00,0.00,0.00}{No}    \\\hline

NDET (Arx~\cite{DBLP:journals/pvldb/PoddarBP19} or Always Encrypt~\cite{DBLP:conf/sigmod/AntonopoulosASE20}) & 2 & 2 or 3 & \textcolor[rgb]{0.00,0.50,0.00}{Yes} & \textcolor[rgb]{0.00,0.50,0.00}{Yes} & \textcolor[rgb]{.00,0.00,0.00}{No} & \textcolor[rgb]{1.00,0.00,0.00}{No}    \\\hline

Indexable-SSE (PB~\cite{DBLP:journals/pvldb/LiLWB14}- or IB~\cite{DBLP:conf/icde/LiL17}-Tree) & 4 & 2 & \textcolor[rgb]{1.00,0.00,0.00}{No}$^\ast$ & \textcolor[rgb]{1.00,0.00,0.00}{No}   & \textcolor[rgb]{1.00,0.00,0.00}{No}  & \textcolor[rgb]{1.00,0.00,0.00}{No} \\\hline

Indexable-SSE with ORAM (Blind Seer~\cite{DBLP:conf/sp/PappasKVKMCGKB14}) & 4 & 4 & \textcolor[rgb]{1.00,0.00,0.00}{No}$^\ast$ & \textcolor[rgb]{1.00,0.00,0.00}{No}   & \textcolor[rgb]{1.00,0.00,0.00}{No}  & \textcolor[rgb]{1.00,0.00,0.00}{Yes} \\\hline

Non-indexable-SSE~\cite{DBLP:conf/sp/SongWP00,DBLP:journals/jcs/CurtmolaGKO11}  & 2 & 4 & \textcolor[rgb]{1.00,0.00,0.00}{No}  & \textcolor[rgb]{1.00,0.00,0.00}{No}   & \textcolor[rgb]{1.00,0.00,0.00}{No}  & \textcolor[rgb]{1.00,0.00,0.00}{No} \\\hline

SGX system (Opaque~\cite{DBLP:conf/nsdi/ZhengDBPGS17})   &  2  &  3 & \textcolor[rgb]{1.00,0.00,0.00}{No} &  \textcolor[rgb]{1.00,0.00,0.00}{No} &  \textcolor[rgb]{1.00,0.00,0.00}{No} & \textcolor[rgb]{1.00,0.00,0.00}{No} \\\hline

MPC or SS (Jana~\cite{DBLP:journals/iacr/ArcherBLKNPSW18}) &  4  &  4
& \textcolor[rgb]{1.00,0.00,0.00}{No} & \textcolor[rgb]{1.00,0.00,0.00}{Yes}  & \textcolor[rgb]{1.00,0.00,0.00}{No} & \textcolor[rgb]{0.00,0.50,0.00}{Yes}  \\\hline

\textbf{\textsc{Concealer}} &  1  &  1 & \textcolor[rgb]{0.00,0.50,0.00}{Yes} & \textcolor[rgb]{0.00,0.50,0.00}{Yes}   & \textcolor[rgb]{0.00,0.50,0.00}{Yes}  & \textcolor[rgb]{0.00,0.50,0.00}{Yes (partial)}\\\hline
\end{tabular}
\caption{Comparing different techniques vs \textsc{Concealer}. Note: 1: Very fast, 2: Fast, 3: Slow. 4: Very slow. $\ast$: Indexable SSEs build their own indexes and their index traversal techniques are not in-built in existing commercial DBMS.}
\label{table:compare_crypto}
\BBB\BBB
\end{table*}
\egroup

We consider the problem wherein trusted data producers ($\mathcal{DP}$) share users' spatial time-series data in the encrypted form with untrusted service providers ($\mathcal{SP}$) to empower $\mathcal{SP}$ to build value-added applications for users. Examples include a cellular company sharing data about the cell tower a user's mobile phone is connected to, or a map service (\textit{e}.\textit{g}., Google Map) sharing the user's GPS coordinates with a third-party providing location-based applications. Another example is an organization/university providing  WiFi connectivity data about the access point a user's device is connected to, for applications such as building dynamic occupancy maps \cite{blincomapny}. We classify applications supported by $\mathcal{SP}$ using user's data into two classes:

\begin{enumerate}[leftmargin=0.2in]
    \item 
    \textbf{Aggregate Applications} that aggregate data of multiple users to build novel applications. Examples include occupancy of different regions, heat maps, and count of distinct visitors to a given region over a period of time. Such applications are already supported by several service providers, \textit{e}.\textit{g}., Google Maps supports information about busy-status and wait times at stores such as restaurants and shopping malls.

\item 
\textbf{Individualized Applications} that allow users to ask queries based on their own past movements, \textit{e}.\textit{g}., locations a person spent the most time during a given time interval, finding the number of people came in contact with, and/or other aggregate operations on user's data. Such applications can be very useful for several contexts including exposure tracing in the context of infectious diseases~\cite{DBLP:journals/corr/abs-2005-12045}.
\end{enumerate}



Implementing applications at the (untrusted) $\mathcal{SP}$ requires: (\textit{i}) $\mathcal{DP}$ to appropriately encrypt data prior to sharing with $\mathcal{SP}$, (\textit{ii}) $\mathcal{SP}$ to be able to execute queries on behalf of the user over the encrypted data, and (\textit{iii}) the user to be able to decrypt the encrypted answers returned by $SP$. {\color{black}Realizing such a data-sharing architecture leads to the following three requirements, (of which the first two are
relatively straightforward, while the third requires a careful design of a new cryptographic technique that this paper focuses on):}







\medskip
\noindent
\textbf{R1: Query formulation by the user.} Given that data is encrypted by $\mathcal{DP}$ and is hosted at $\mathcal{SP}$, the user needs to formulate the query to enable $\mathcal{SP}$  to execute it over encrypted data. The users can formulate an appropriate encrypted query, if they know the key used for encryption by $\mathcal{DP}$. However, $\mathcal{DP}$ cannot share the key with users to prevent them from decrypting the entire dataset. A trivial way to overcome this problem is to involve $\mathcal{DP}$ in processing queries. Particularly, a user can submit queries to $\mathcal{DP}$ that converts the query into appropriate trapdoors to be executed on encrypted data at $\mathcal{SP}$, fetches the partial results from $\mathcal{SP}$, and processes the fetched rows, before producing the final answer to users. Such an architecture incurs significant overhead at $\mathcal{DP}$, which essentially requires $\mathcal{DP}$ to mediate every user query, essentially pushing them to act as a surrogate $\mathcal{SP}$. Thus, \emph{\textbf{the first {\color{black}requirement} is how users can formulate appropriate encrypted queries {\bf without} involving $\boldsymbol{\mathcal{DP}}$ in query processing}}.

{\color{black}We can overcome this requirement \textit{trivially} by using
secure hardware} (\textit{e}.\textit{g}., Intel Software Guard eXtensions\footnote{\scriptsize
Recent Intel CPUs introduced SGX that allows us to create a small trusted execution environment, called \emph{enclave} that is isolated and protected from the rest of the system. SGX protects computations from the operating system (controlled by the third-party) and from numerous applications/system-level attacks. Unfortunately, existing implementations of SGX are prone to side-channel attacks that exploit one of the microarchitectural components of CPUs, \textit{e}.\textit{g}., cache-lines, branch execution, page-table access~\cite{DBLP:conf/ccs/WangCPZWBTG17,DBLP:conf/uss/0001SGKKP17,DBLP:journals/corr/abs-1811-05378}, and power attacks. Nevertheless, systems T-SGX~\cite{t-sgx} and Sanctum~\cite{Sanctum} have evolved to overcome such attacks, and it is believed that future versions of SGX will be resilient to those attacks.} (SGX)~\cite{sgx}) at $\mathcal{SP}$ that works as a trusted agent of $\mathcal{DP}$. SGX receives queries encrypted using the public key of SGX (which we assume to known to all) from users, decrypts the query, converts the query into appropriate secure trapdoors, and provides the answer.

\medskip
\noindent
\textbf{R2: Preventing $\mathcal{SP}$ from impersonating a user.} Since we do not wish to involve $\mathcal{DP}$ during query processing, all users ask queries directly to  $\mathcal{SP}$. Such query representations should not empower  $\mathcal{SP}$ to mimic/masquerade as a legitimate user to gain access to the cleartext data from the answers to the query. Thus, \emph{\textbf{the second {\color{black}requirement} is how will the system prevent $\boldsymbol{\mathcal{SP}}$ to mimic as a user and execute a query}}.

{\color{black}We overcome this requirement \textit{trivially} by building a list of \emph{registered users}}, who are allowed to execute queries on the encrypted data (after a negotiation between users and $\mathcal{DP}$) at $\mathcal{DP}$ and provides it in encrypted form to $\mathcal{SP}$. The registry contains credential information (\textit{e}.\textit{g}., public/private key and authentication information of users) about the users who are interested in $\mathcal{SP}$ applications. Thus, before generating any trapdoor by SGX, it first authenticates the user and provides the final answers encrypted using the public key of the user.

\medskip
\noindent
\textbf{R3: Selecting the appropriate encryption technique.}
{\color{black}Spatial time-series data brings in new challenges (as compared to other datasets) in terms of a large amount of the dataset and dynamically arriving data. Also, spatial time-series data show new opportunities in terms of limited types of queries (\textit{i}.\textit{e}., not involving complex operations such as join and nested queries).} Particularly, the data encryption and storage must sustain the data generation rate, \textit{i}.\textit{e}., the encryption mechanism must support dynamic insertion without the high overhead. Further, cryptographic query execution time should scale to millions of records. Finally, the system must support strong security properties such that the ciphertext representation and query execution do not reveal information about the data to $\mathcal{SP}$. Note that ciphertext representation leaks data distribution only when deterministic encryption (DET) is used. Query execution leaks information about data due to search- and access-patterns leakages, and volume/output-size leakage. 
In~\S\ref{subsec:Comparison and Advantages}, we discuss these leakages and argue that none of the existing cryptographic query processing techniques satisfy all the above requirements. Thus, \emph{\textbf{the {\color{black}third requirement} is how to design a system that has efficient data encryption and query execution techniques, and not prone to such leakages}}.



\medskip
\noindent
\textbf{\Large \textsc{Concealer}.}
{\color{black}We design, develop, and implement a secure spatial time-series database, entitled \textsc{Concealer}. This paper focuses on how \textsc{Concealer} addresses the above-mentioned third requirement, and below, we briefly discuss the proposed solution to the requirement R3. In short, \textsc{Concealer} is carefully designed to support a high rate of data arrival, and large data sets, but it only supports a limited nature of spatial time-series queries required by the domain of interest. To a degree, \textsc{Concealer} can be considered more of a vertical technology compared to general-purpose horizontal solutions, which as will be discussed in~\S\ref{subsec:Comparison and Advantages}, lack the ability to support application/data that motivates \textsc{Concealer}.
} 

{\color{black}
\textsc{Concealer}, for fast data encryption and minimum cryptographic overheads on each tuple, uses a variant of deterministic encryption that produces secure ciphertext (that does not reveal data distribution) and is fast enough to encrypt tuples ($\approx$37,185 tuples/min). Further, \textsc{Concealer} exploits the index supported by MySQL.} Note that we do not use any specialized index (\textit{e}.\textit{g}., PB-tree~\cite{DBLP:journals/pvldb/LiLWB14} and IB-Tree~\cite{DBLP:conf/icde/LiL17}) and do not require to build the entire index for each insertion at the trusted side. Since \textsc{Concealer} users an index supported by DBMS, it supports efficient query execution. For a point query on 136M rows, \textsc{Concealer} needs at most 0.9s. Thus, our implementation of DET and the use of indexes supported by DBMS satisfy the requirements of fast data insertion and fast query execution.

To address the security challenge during query execution, \textsc{Concealer}  (\textit{i}) prevents output-size by fixing the unit of data retrieval of the form of bins, formed over the tuples of a given time period; care is taken to ensure that each bin must be of identical size (by implementing a variant of bin-packing algorithm~\cite{Coffman:1996:AAB:241938.241940}), and (\textit{ii}) hides \emph{partial} access-patterns, due to retrieving a fixed bin having different tuples corresponding to different sensor readings (with different location/time/other values) for any query corresponding to the element of the bin. That means the adversary observes: which fixed tuples are fetched for a set of queries including the real query posed by the user. However, the adversary cannot find which of the fetched tuples satisfy the user query. Since our focus is on practical system implementation, we relax the complete access-pattern hiding requirements. The exact security offered by \textsc{Concealer} will be discussed in~\S\ref{sec:Security Threats and Properties}.

Since we fetch a bin of several tuples, to filter the useless tuples that do not meet the query predicates, SGX at $\mathcal{SP}$ filters them, (while also hides complete access-patterns inside SGX by performing oblivious operations). Further, to verify the integrity of the data before producing the answer, \textsc{Concealer} provides a \emph{non-mandatory} hash-chain-based verification.

\medskip
\noindent\textbf{Evaluation.}
We evaluate \textsc{Concealer} on a real WiFi dataset collected from at UCI. To evaluate its scalability, we executed the algorithms on 136M rows, the size that previous existing cryptographic techniques cannot support. We also compare \textsc{Concealer} against SGX-based Opaque~\cite{DBLP:conf/nsdi/ZhengDBPGS17}. To the best of our knowledge, there is no system that supports identical security properties (hiding output-size and hiding partial access-patterns, while supporting indexes for efficient processing). Our algorithms can be used to deal with non-time-series datasets also; thus, to evaluate algorithms' practicality, we evaluate aggregation queries on 136M rows LineItem table of TPC-H benchmark.

\subsection{Comparison \& Advantages of \textsc{Concealer}}
\label{subsec:Comparison and Advantages}
We discuss common leakages from cryptographic solutions, argue that they do not satisfy the requirements of fast data insertion, fast query execution, and/or security against leakages (see Table~\ref{table:compare_crypto} for a comparison).

\noindent
\textbf{Leakages.} Cryptographic techniques show the following leakages:

\begin{enumerate}[leftmargin=0.2in]

\item \textit{Data distribution leakage from the storage~\cite{DBLP:conf/ccs/CashGPR15}}: allows an adversary to learn the frequency-count of each value by just observing ciphertext. DET reveals such information.

\item \textit{Search- and access-patterns leakages~\cite{DBLP:conf/ndss/IslamKK12,DBLP:conf/ccs/CashGPR15}}: occur during query execution. Search-pattern leakages allow an adversary to learn if and when a query is executed, while access-patterns leakage allows learning which tuples are retrieved (by observing the (physical) address/location of encrypted tuples) to answer a query. Practical techniques, such as order-preserving encryption (OPE)~\cite{DBLP:conf/sigmod/AgrawalKSX04}, DET, symmetric searchable encryption (SSE)~\cite{DBLP:journals/pvldb/LiLWB14,DBLP:conf/icde/LiL17}, and secure hardware-based techniques~\cite{DBLP:conf/nsdi/ZhengDBPGS17,DBLP:conf/sigmod/AntonopoulosASE20,DBLP:conf/sp/SchusterCFGPMR15,DBLP:conf/uss/DinhSCOZ15,DBLP:conf/sp/PriebeVC18}, reveal the access-patterns. In contrast, non-efficient techniques (\textit{e}.\textit{g}., secret-sharing~\cite{DBLP:journals/iacr/ArcherBLKNPSW18,DBLP:journals/pvldb/BaterEEGKR17} or oblivious RAM (ORAM) based techniques) hide access-patterns.

\item \textit{Volume/output-size leakage~\cite{DBLP:conf/ccs/CashGPR15}}: allows an adversary (having some background knowledge) can deduce the data by simply observing the size of outputs (or the number of qualifying tuples).~\cite{DBLP:conf/ndss/IslamKK12,DBLP:conf/ccs/CashGPR15,DBLP:journals/iacr/Naveed15} showed that output-size may also leak data distribution. \emph{Access-patterns revealing techniques implicitly disclose the output-size}. Moreover, the seminal work~\cite{DBLP:conf/ccs/KellarisKNO16} showed that the output-size revealed even due to access-pattern hiding techniques enables the attacker to reconstruct the dataset.
A possible solution is adding fake tuples with the real data, thereby each value has an identical number of tuples and using indexable SSEs. However,~\cite{DBLP:journals/iacr/Naveed15} showed that it will be even more expensive than simply scanning the entire database in SGX (or download the data at $\mathcal{DP}$ to execute the query locally). Existing output-size preventing solutions, \textit{e}.\textit{g}., Kamara et al.~\cite{DBLP:conf/eurocrypt/KamaraM19} or Patel et al.~\cite{DBLP:conf/ccs/PatelPYY19}, suffer from one major problem:~\cite{DBLP:conf/eurocrypt/KamaraM19} fetches $\alpha\times \mathit{max}$, $\alpha>2$, rows, while~\cite{DBLP:conf/ccs/PatelPYY19} fetches $2\times \mathit{max}$ rows with additional secure storage of some rows (which is the function of DB size), where $\mathit{max}$ is the maximum number of rows a value can have. Thus, both~\cite{DBLP:conf/eurocrypt/KamaraM19} and~\cite{DBLP:conf/ccs/PatelPYY19}  fetch more than the desired rows, \textit{i}.\textit{e}., $\mathit{max}$. Moreover, both~\cite{DBLP:conf/eurocrypt/KamaraM19} and~\cite{DBLP:conf/ccs/PatelPYY19} cannot deal with dynamic data.
\end{enumerate}


\medskip
\noindent\textbf{Existing techniques in terms of data insertion, query execution, and leakages.}
Existing encrypted search techniques differ in their support for dynamic data, efficient query execution, and offered security properties. For instance, DET supports very efficient insertion and query processing, while its ciphertext data leaks data distribution.

Non-indexable techniques/systems (\textit{e}.\textit{g}., SSE~\cite{DBLP:conf/sp/SongWP00,DBLP:journals/jcs/CurtmolaGKO11}, secret-sharing (SS)~\cite{DBLP:journals/iacr/ArcherBLKNPSW18,DBLP:journals/pvldb/BaterEEGKR17}, secure hardware-based systems~\cite{DBLP:conf/nsdi/ZhengDBPGS17}) allow fast data insertions by just encrypting the data, but have inefficient query response time, due to unavailability of an index, and hence, reading the entire data. SS hides search- and access-patterns, while others reveal. Moreover, all such techniques are prone to output-size leakage. 

In contrast, indexable techniques/systems (\textit{e}.\textit{g}., indexable SSEs (such as PB-Tree~\cite{DBLP:journals/pvldb/LiLWB14}, IB-Tree~\cite{DBLP:conf/icde/LiL17}) and secure hardware-based index~\cite{DBLP:conf/sp/MishraPCCP18})  have faster query execution, but show slow data insertion rate, due to building the entire index at the trusted side for each data insertion; \textit{e}.\textit{g}.,~\cite{DBLP:conf/sp/PappasKVKMCGKB14} showed that creating a secure index over 100M rows took more than 1 hour. Moreover, these indexable techniques use \textbf{\emph{specialized indexes}} that require specialized encryption and tree traversal protocols that are not supported in the existing standard database systems. This, in turn, limits their usability in dealing with large-scale time-series datasets. All such indexable solutions reveal output-size. While indexable solutions mixed with ORAM (\textit{e}.\textit{g}.,~\cite{DBLP:conf/sp/PappasKVKMCGKB14}) hide search- and access-patterns, they are not efficient for query processing (due to several rounds of interaction between the data owner and the server to answer a query). {\color{black} In summary, spatial time-series data adds complexity since (\textit{i}) it can be very large, and (\textit{ii}) arrives dynamically (possibly a high velocity). Existing techniques, as discussed above, are not suitable to support secure data processing over such data.}

\medskip
\noindent
\textbf{Advantages of \textsc{Concealer}.}
{\color{black}(\textit{i}) {\emph{Frequent data insert.}} We deal with frequent bulk data insertions (which is a requirement of spatial time-series datasets).}
(\textit{ii}) {\emph{Deal with large-size data.}} We handle large-sized data with several attributes and large-sized domain efficiently, as our experimental results will show in~\S\ref{sec:Experimental Evaluation}.
(\textit{iii}) {\emph{Output-size prevention.}} While \textsc{Concealer} satisfies the standard security notion (supported by existing SSEs), \textit{i}.\textit{e}., indistinguishability under chosen keyword attacks (IND-CKA)~\cite{DBLP:journals/jcs/CurtmolaGKO11}, it also prevents output-size attacks, unlike IND-CKA.
(\textit{iv}) {\emph{Oblivious processing in SGX.}} As we use the current SGX architecture, suffering from side-channel attacks (\textit{e}.\textit{g}., cache-line, branch shadow, and page-fault attack~\cite{DBLP:conf/ccs/WangCPZWBTG17,DBLP:conf/uss/0001SGKKP17,DBLP:journals/corr/abs-1811-05378}) that enable the adversary to deduce information based on access-patterns in SGX. Thus, we incorporate techniques to deal with these attacks.

\subsection{Scoping the Problem}
There are other aspects, for them either solutions exist or this paper does not deal with them, as:
(\textit{i}) \noindent\textit{Key management.} We do not focus on building/improving key infrastructure for public/private keys, as well as, key generation and sharing between SGX and $\mathcal{DP}$. Further, changing the keys of encrypted data and re-encrypting the data is out of the scope of this paper, though one may use the recent approach~\cite{DBLP:conf/ccs/JareckiKR19} to do so. Also, we do not focus on SGX remote attestation.
(\textit{ii})
\noindent\textit{Man-in-the-middle (MiM) or replay attacks.} There could be a possibility of MiM and replay attacks on SGX during attestation and query execution. We do not deal with both issues, and techniques~\cite{DBLP:conf/codaspy/DharPKC20} can be used to avoid such attacks.
(\textit{iii})
\noindent\textit{Inference from the number of rows.} Since we send data in epochs, different numbers of tuples in different epochs (\textit{e}.\textit{g}., epochs for day vs night time) may reveal information about the user. This can be prevented by sending the same number of rows in each epoch (equals to the maximum rows in any epoch). The current implementation of \textsc{Concealer} does not deal with this issue.
(\textit{iv})
\noindent\textit{Inference from occupancy count.} Occupancy information mixed with background knowledge reveals the presence/absence of a person at a location (\textit{e}.\textit{g}., offices). We do not deal with these inferences, and differential privacy techniques mixed with SGX~\cite{xuhermetic,DBLP:journals/pvldb/BaterHEMR18} can be used to deal with such issues.

\begin{figure*}[!t]

  \centering
  \includegraphics[scale=0.8]{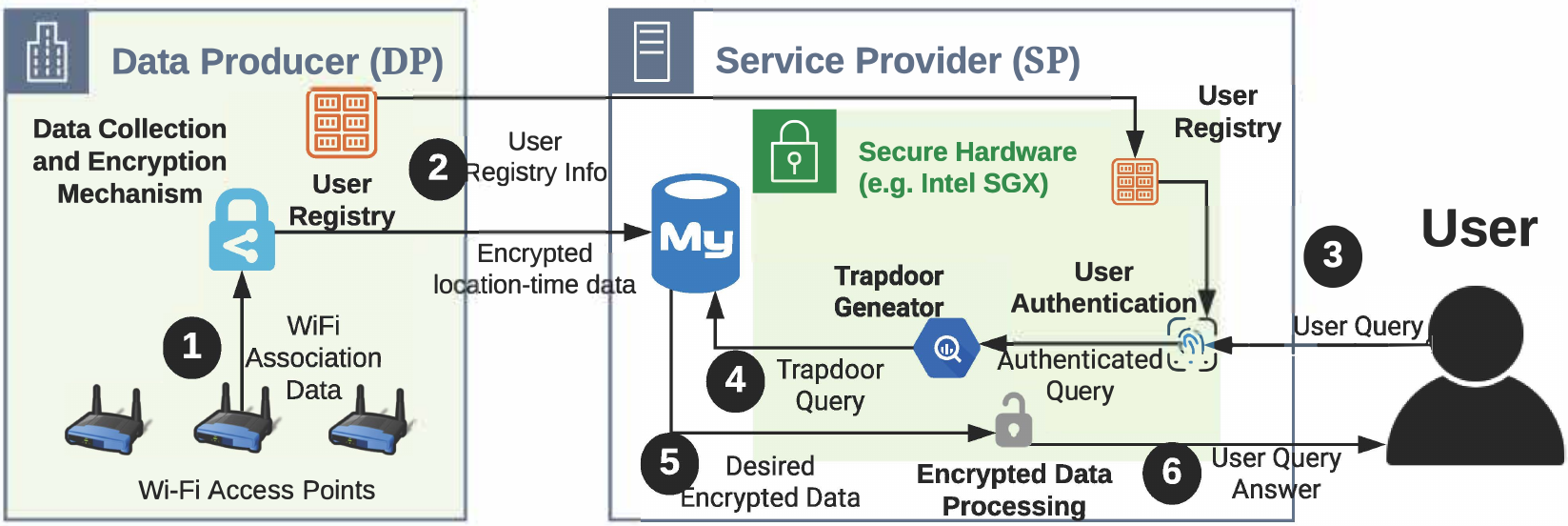}
\B
\caption{\textsc{Concealer} model.}
\BB
  \label{fig:Concealer model}
\end{figure*}

\subsection{Outline of the Paper}
The other sections of the papers are organized as follows:
\begin{enumerate}[leftmargin=0.2in]
    \item \S\ref{sec:securetimedb architecture} provides an overview of the entities involved in \textsc{Concealer}, its architecture, and a high-level description of the proposed algorithms.
\item \S\ref{sec:Data encryption and Outsourcing} provides the details of data encryption and data outsourcing algorithms.

\item \S\ref{sec:point Query Execution} and ~\S\ref{sec:Range Query} provide the details of algorithms for point queries and range queries, respectively.

\item \S\ref{sec:Insert Operation} provides the algorithm for dynamic data insertion, and~\S\ref{sec:Preventing Attacks due to Query Workload} provides an algorithm to deal with information leakages due to the query workload.

\item \S\ref{sec:Security Threats and Properties} provides the security properties satisfies by \textsc{Concealer} and discusses the information leakages due to different query execution algorithms.

\item \S\ref{sec:Experimental Evaluation} provides experimental results of \textsc{Concealer} and compares them against different cryptographic techniques.

\end{enumerate}

\section{\textsc{Concealer} Overview}
\label{sec:securetimedb architecture}
This section provides an overview of entities involved in \textsc{Concealer} and its architecture with a high-level overview of algorithms.

\subsection{Entities and Assumptions}
\label{subsec:Entities and Assumptions}
\textsc{Concealer} consists of the following three entities:

\begin{itemize}[leftmargin=0.15in]
    \item 
\textbf{Data provider} $\boldsymbol{\mathcal{DP}}$\textbf{:}  is a trusted entity that collects user's spatial time-series data as part of its regular operation (\textit{e}.\textit{g}., providing cellular service to users). ${\mathcal{DP}}$ shares such data in encrypted form with service providers ${\mathcal{SP}}$. ${\mathcal{DP}}$, also, maintains a \emph{registry}, one per ${\mathcal{SP}}$, that contains a list of identification information of users, who have registered to use the application provided by the corresponding ${\mathcal{SP}}$ (\textit{i}.\textit{e}., can run queries at that ${\mathcal{SP}}$). As will be clear,  this registry helps to restrict the users to request individualized applications about other users.


\item 
 \textbf{Service provider} $\boldsymbol{\mathcal{SP}}$\textbf{:} is an untrusted entity that develops location-based applications (as mentioned in~\S\ref{sec:Introduction}) over encrypted data. To do so, $\mathcal{SP}$ hosts  secure hardware, SGX, that works as a trusted agent of $\mathcal{DP}$.\footnote{{\scriptsize The assumption of secure hardware at untrusted third-party machines is consistent with emerging system architectures; \textit{e}.\textit{g}., Intel machines are equipped with SGX~\cite{url1}.}} SGX and $\mathcal{DP}$ share a secret key $s_k$ (used for encryption/ decryption of data), and this key is unknown to all other entities.

An untrusted $\mathcal{SP}$ may try to learn user's data passively by either observing the data retrieved by SGX or exploiting side-channel attacks on SGX during query execution. It may further learn user's data by actively injecting the fake data into the database and then observing the corresponding ciphertext and query access-patterns. We assume that $\mathcal{SP}$ knows background information, \textit{e}.\textit{g}., metadata, the schema of the relation, the number of tuples, and the domain of attributes. However, the adversarial $\mathcal{SP}$ cannot alter anything within the secure hardware and cannot decrypt the data, due to the unavailability of the encryption key. Such assumptions are similar to those considered in the past work related to SGX-based computation~\cite{DBLP:conf/nsdi/ZhengDBPGS17,DBLP:conf/sp/SchusterCFGPMR15,DBLP:conf/uss/DinhSCOZ15,DBLP:conf/sp/PriebeVC18}, work on attacks based on background knowledge in~\cite{DBLP:conf/ccs/NaveedKW15,DBLP:conf/ccs/KellarisKNO16}, and location-based services~\cite{infocom}. 

\item 
\textbf{User or data consumer} $\boldsymbol{\mathcal{U}}$\textbf{:} that uses the services of $\mathcal{DP}$ (such as cellular or WiFi connectivity) and queries to $\mathcal{SP}$. We assume that $\mathcal{U}$ have their public and private keys, which are used to authenticate $\mathcal{U}$ at $\mathcal{SP}$ (via SGX against registry). As mentioned in~\S\ref{sec:Introduction}, ${\mathcal{U}}$ can request both aggregate and individualized queries. While ${\mathcal{U}}$ is trusted with the data that corresponds to themselves, they are not trusted with data belonging to other users. Finally, we assume that while ${\mathcal{U}}$ can execute the aggregation queries, they do not collude with ${\mathcal{SP}}$, \textit{i}.\textit{e}., they do not share cleartext results of any query with ${\mathcal{SP}}$.
\end{itemize}



\bgroup
\def\arraystretch{1.4}
\begin{table*}[!t]
\BBB\BBB
    \begin{subtable}[b]{.15\linewidth}
      \centering
      \scriptsize
                \begin{tabular}{|l||l|l|l|}
                  \hline
        & $\mathcal{L}$ & $\mathcal{T}$ & $\mathcal{O}$  \\\hline
   $r_1$  & $l_1$ & $t_1$ & $o_1$ \\\hline
   $r_2$  & $l_1$ & $t_2$ & $o_2$ \\\hline
   $r_3$  & $l_2$ & $t_3$ & $o_2$ \\\hline
   $r_4$  & $l_1$ & $t_4$ & $o_1$ \\\hline
   $r_5$  & $l_2$ & $t_5$ & $o_3$ \\\hline
   $r_6$  & $l_3$ & $t_6$ & $o_2$ \\\hline
\end{tabular}
\B
\caption{A relation $R$ in cleartext at $\mathcal{DP}$.}
\label{tab:Cleartext data at the data provider}
    \end{subtable}\quad
        \begin{subtable}[b]{.37\linewidth}
    \begin{center}
    \includegraphics[scale=0.4]{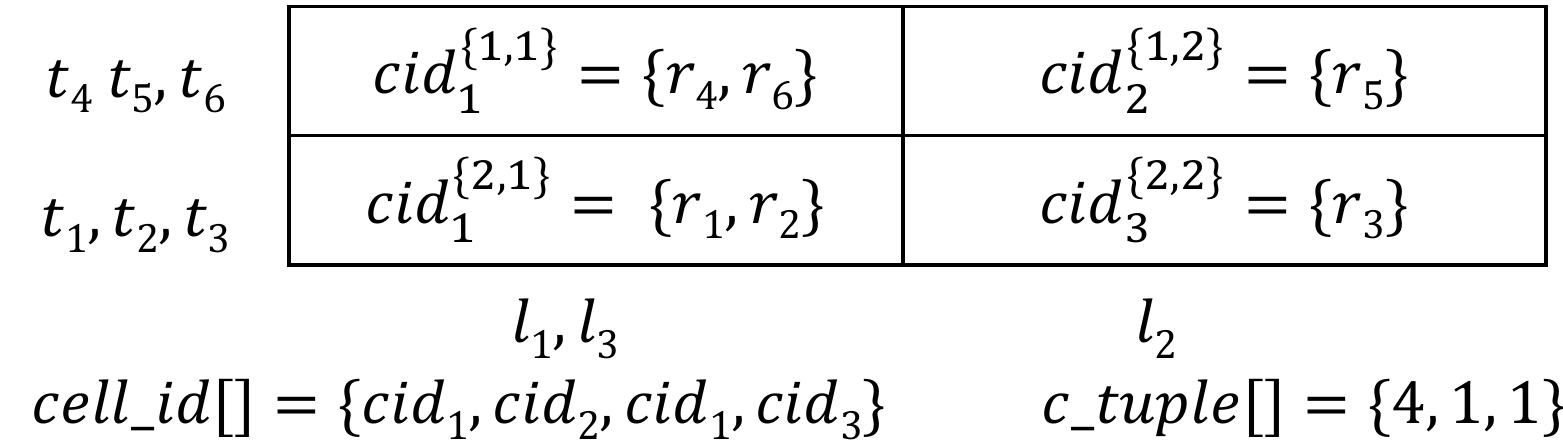}
    \end{center}
    \caption{The grid created at $\mathcal{DP}$ for rows of Table~\ref{tab:Cleartext data at the data provider}.}
\label{tab:the grid 2times2}
    \end{subtable}
    \begin{subtable}[b]{.44\linewidth}
    \begin{center}
\scriptsize
\begin{tabular}{|l||l|l|l|l|l|l}
  \hline
          & $\mathcal{L}$ &  $\mathcal{O}$ &
          Tuple &
          $\mathrm{Index}(\mathcal{L},\mathcal{T})$
          \\\hline

   $r_1^{\prime}$  & $\mathcal{E}_k(l_1||t_1)$ & $\mathcal{E}_k(o_1||t_1)$  & $\mathcal{E}_k(l_1||t_1||o_1)$ & $\mathcal{E}_k(\mathit{cid}_1||1)$   \\\hline

   $r_7^{\prime}$  & $\mathcal{E_{\mathit{nd}}}(\mathit{fake})$ & $\mathcal{E_{\mathit{nd}}}(\mathit{fake})$  & $\mathcal{E_{\mathit{nd}}}(\mathit{fake})$ & $\mathcal{E}_k(\mathit{f}||1)$   \\\hline

   $r_2^{\prime}$  & $\mathcal{E}_k(l_1||t_2)$ &  $\mathcal{E}_k(o_2||t_2)$ & $\mathcal{E}_k(l_1||t_2||o_2)$ & $\mathcal{E}_k(\mathit{cid}_1||2)$    \\\hline

   $r_3^{\prime}$  & $\mathcal{E}_k(l_2||t_3)$ &  $\mathcal{E}_k(o_2||t_3)$ & $\mathcal{E}_k(l_2||t_3||o_2)$   & $\mathcal{E}_k(\mathit{cid}_3||1)$    \\\hline

   $r_8$  & $\mathcal{E_{\mathit{nd}}}(\mathit{fake})$ & $\mathcal{E_{\mathit{nd}}}(\mathit{fake})$  & $\mathcal{E_{\mathit{nd}}}(\mathit{fake})$ & $\mathcal{E}_k(\mathit{f}||2)$   \\\hline

   $r_4^{\prime}$  & $\mathcal{E}_k(l_1||t_4)$ &  $\mathcal{E}_k(o_1||t_4)$ & $\mathcal{E}_k(l_1||t_4||o_1)$    & $\mathcal{E}_k(\mathit{cid}_1||3)$    \\\hline




  $r_5^{\prime}$  & $\mathcal{E}_k(l_2||t_5)$ &  $\mathcal{E}_k(o_3||t_5)$ & $\mathcal{E}_k(l_2||t_5||o_3)$    & $\mathcal{E}_k(\mathit{cid}_2||1)$    \\\hline


   $r_6^{\prime}$  & $\mathcal{E}_k(l_3||t_6)$ &  $\mathcal{E}_k(o_2||t_6)$ & $\mathcal{E}_k(l_3||t_6||o_2)$    & $\mathcal{E}_k(\mathit{cid}_1||4)$    \\\hline
\end{tabular}

$\mathit{Ecell\_id}[2,2]=\mathcal{E}_{\mathit{nd}}(\{\mathit{cid}_1,\mathit{cid}_2,\mathit{cid}_1,\mathit{cid}_3\})$

$\mathit{Ec\_tuple}[3]=\mathcal{E}_{\mathit{nd}}(\{4,1,1\})$
\end{center}
\BB
\caption{Encrypted data with encrypted counters at $\mathcal{SP}$.}
\label{tab:Encrypted data at the cloud}
   \end{subtable}
   \BBB\B
     \caption{Input time-series relation and output of data encryption algorithm.}
\BBB\BBB
\end{table*}
\egroup

\subsection{Architecture}
\label{subsec:Architecture}
\textsc{Concealer} (Figure~\ref{fig:Concealer model}) consists of the following phases: 

\medskip
\noindent\textbf{\textsc{Phase 0:} \emph{Preliminary step: Announcement of $\boldsymbol{\mathcal{SP}}$ by $\boldsymbol{\mathcal{DP}}$.}} As a new $\mathcal{SP}$ is added into the system, $\mathcal{DP}$ announces about the $\mathcal{SP}$ to all their users. Only interested users inform to ${\mathcal{DP}}$ if they want to use $\mathcal{SP}$'s application. Information of such users, their device-id, and authentication information is stored by ${\mathcal{DP}}$ in the registry.

\medskip
\noindent\textbf{\textsc{Phase 1:} \emph{Data upload by $\boldsymbol{\mathcal{DP}}$.}} $\mathcal{DP}$ collects spatial time-series data of the form $\langle l_i,t_i,o_i\rangle$ (\encircle{1}), where $l_i$ is the location, $t_i$ is the time, and $o_i$ is the observed value at $l_i$ and $t_i$. 
For instance, in the case of WiFi data, the location may correspond to the region covered by a specific WiFi access-point, and the observation corresponds to a particular device-id connected to that access-point at a given time. $\mathcal{DP}$ encrypts the data using the mechanism given in~\S\ref{sec:Data encryption and Outsourcing} and provides the encrypted data to $\mathcal{SP}$ (\encircle{2}) along with encrypted registry and verifiable tags (to verify the data integrity at $\mathcal{SP}$ by SGX).

\textsc{Concealer} considers the data as a relation $R$ with three attributes: $\mathcal{T}$ (time), $\mathcal{L}$ (location), and $\mathcal{O}$ (observation). 
Table~\ref{tab:Cleartext data at the data provider} shows an example of the cleartext spatial time-series data, (which will be used in this paper to explain \textsc{Concealer}). In Table~\ref{tab:Cleartext data at the data provider}, we have added a row-id $r_i$ ($1\leq i\leq 6$) to refer to individual rows of Table~\ref{tab:Cleartext data at the data provider}. Table~\ref{tab:Encrypted data at the cloud} shows an example of the encrypted spatial time-series data as the output of \textsc{Concealer}.


\medskip
\noindent\textbf{\textsc{Phase 2:} \emph{Query generation at $\boldsymbol{\mathcal{U}}$.}} A query
$Q= \langle \mathit{qa}, \mathit{att}\rangle$, where $\mathit{qa}$ is an aggregation (count, maximum, minimum, top-k, and average) or selection operation for a given condition, and $\mathit{att}$ is a set of attributes with predicates on which query will be executed, is submitted to $\mathcal{SP}$ (\encircle{3}).  $\mathit{qa}$ is always encrypted to prevent $\mathcal{SP}$ to know the query values.

\medskip
\noindent\textbf{\textsc{Phase 3:} \emph{Query processing at $\boldsymbol{\mathcal{SP}}$.}} $\mathcal{SP}$ holds encrypted spatial time-series dataset and the user query (submitted to the secure hardware SGX). SGX, first, authenticate the user, and then, translates the query into a set of appropriate secured query trapdoors to fetch the tuples from the databases (\encircle{4}). Note that since the individualized application is executed for the user itself, trapdoors are only generated if the authentication process succeeds to find that the user is wishing to know his past behavior.

The trapdoors are generated by following the methods of~\S\ref{sec:point Query Execution} for point queries or the method of~\S\ref{sec:Range Query} for range queries. On receiving encrypted tuples from the database  (\encircle{5}), the secure hardware, first, \emph{optionally} checks their integrity using verifiable tags, and if find they have not tampered, decrypts them, if necessary, and obliviously processes them to produce the final answer to the user  (\encircle{6}).

\medskip
\noindent\textbf{\textsc{Phase 4:} \emph{Answer decryption at $\boldsymbol{\mathcal{U}}$.}} On receiving the answer, $\mathcal{U}$ decrypts them.

\subsection{Algorithm Overview}
\label{subsec:High-Level Overview of Protocols}
Before going into details of \textsc{Concealer}'s data encryption and query execution algorithms, we first explain them at the high-level.

\medskip
\noindent\textbf{Data encryption method at $\boldsymbol{\mathcal{DP}}$:}  partitions the time into slots, called \emph{epochs}, and for each epoch, it executes the encryption method that consists of the following three stages:

\smallskip
\noindent
\underline{\textsc{Stage} 1: \textit{Setup.}} Assume that we want to deal with two attributes ($A$ and $B$), (\textit{e}.\textit{g}., location and time). This stage:
(\textit{i}) creates a grid of size, say $x\times y$,
(\textit{ii}) sub-partitions the time into $y$ subintervals, \textit{e}.\textit{g}., for an epoch of 9-10am, creates $y$ subintervals as: 9:00-9:10, 9:11-9:20, and so on, and
(\textit{iii}) using a hash function, say $\mathbb{H}$, allocates $x$ values of $A$ attributes over $x$ columns, allocates $y$ values (or $y$ subintervals) of $B$ attribute to $y$ rows, and allocates some \textbf{\emph{cell-ids}} $< x\times y$ (each with their \textbf{\emph{counters}} initialized to zero) over the grid cells.{ (Such grid-creation steps can be used for more than two columns trivially and extended for non-time-series dataset.)}

\smallskip
\noindent
\underline{\textsc{Stage} 2: \textit{Encryption:}} In this stage, each sensor reading is encrypted and a \emph{verifiable tag} is produced for integrity verification, as:
(\textit{i}) a tuple $t_i$ is allocated to a grid cell corresponding to its desired column (\textit{e}.\textit{g}., location and time) values using a \emph{hash function}, the counter value of the cell-id is increased by one and attached with the tuples, and the tuples is encrypted to produce secure ciphertext with the encrypted counter value as a new attribute value,
(\textit{ii}) a hash-chain is created over the encrypted tuple values of the same cell-id for integrity verification (and verify false data injection or data deletion by $\mathcal{SP}$, and
(\textit{iii}) encrypted fake tuples are added (to prevent output-size leakage at $\mathcal{SP}$).

\smallskip
\noindent
\underline{\textsc{Stage} 3: \textit{Sharing:}} This stage sends encrypted real and fake tuples with encrypted verifiable tags and encrypted cell-id, counter information to $\mathcal{SP}$.

\smallskip
\emph{\textbf{Example.}} Table~\ref{tab:Cleartext data at the data provider} shows six cleartext rows of an epoch. A $2\times 2$ grid with three cell-ids $\mathit{cid}_1$, $\mathit{cid}_2$, and $\mathit{cid}_3$ is shown in Table~\ref{tab:the grid 2times2}. Six cleartext rows are distributed over different cells of the grid. Table~\ref{tab:Encrypted data at the cloud} shows the output of the encryption algorithm with fake tuples to prevent the output-size attack at $\mathcal{SP}$ and an index column created over the cell-ids. Encrypted Table~\ref{tab:Encrypted data at the cloud} with counters and cell-ids (written below Table~\ref{tab:the grid 2times2}) in encrypted form is given to $\mathcal{SP}$.  \textbf{\textit{Details of the encryption method and example will be given in~\S\ref{sec:Data encryption and Outsourcing}.}}$\blacksquare$





\medskip
\noindent\textbf{Data insertion into DBMS at $\boldsymbol{\mathcal{SP}}$:} On receiving the encrypted data from $\mathcal{DP}$, $\mathcal{SP}$ inserts the data into DBMS that creates/modifies the \emph{index based on the counters} associated with each tuple. 

\medskip
\noindent\textbf{Query execution at $\boldsymbol{\mathcal{SP}}$:} 
%
as a pre-processing, the enclave at $\mathcal{SP}$, first, authenticates the user, as mentioned in \textsc{Phase} 3 of~\S\ref{subsec:High-Level Overview of Protocols}, and then, executes the query, as follows:

\smallskip
\noindent
\textit{\underline{Point queries.}} Consider a query on a location $l$ and time $t$. For answering this, the enclave at $\mathcal{SP}$ executes the following steps:
(\textit{i}) first execute the hash function $\mathbb{H}$ on query predicate $l$ and $t$ to know the cell-id, say $\mathit{cid}_z$, that was allocated by $\mathcal{DP}$ to $l$ and $t$,
(\textit{ii}) using the information of cell-id and counter information, which was sent by $\mathcal{DP}$, create \emph{{static bins of a fixed size}} (to prevent output-size leakage),
(\textit{iii}) among the created bins, find a bin, say $B_i$ that has the cell-id $\mathit{cid}_z$ that was obtained in the first step above, and
(\textit{iv}) fetch data from DBMS corresponding to the bin $B_i$, and
(\textit{v}) verify the integrity of data (if needed), \emph{obliviously} process the data against query predicate in the enclave, and decrypt only the desired data.

\smallskip
\noindent
\textit{\underline{Range queries.}} A range query, of course, can be executed by following the above point query method by converting the range query into several point queries. However, to avoid the overhead of several point queries, we create static bins of fixed size over the fixed-sized groups of subintervals and fetch such bins to answer the query by following point queries' step (v). Following\textbf{\textit{~\S\ref{sec:Data encryption and Outsourcing},\S\ref{sec:point Query Execution},\S\ref{sec:Range Query} will describe these algorithms in details}}, and then~\S\ref{sec:Experimental Evaluation} will compare these algorithms on different datasets and against different systems.

\emph{\textbf{Example.}} Underlying DBMS at $\mathcal{SP}$ creates an index over $\mathrm{Index}$ column of Table~\ref{tab:Encrypted data at the cloud}. SGX creates two bins over cell-ids's as $B_1:\langle \mathit{cid}_1\rangle$, $B_2:\langle \mathit{cid}_2, f||1, f||2\rangle$. Note that both bins corresponds to four rows---$B_1$ will fetch $r_1,r_2,r_4,r_6$, and $B_2$ will fetch $r_3,r_5,f||1,f||2$ rows. Thus, the output size will be the same. Now, consider a query $Q=\langle$count, $(l_2,t_5)\rangle$ over Table~\ref{tab:Encrypted data at the cloud}. Here, SGX will know that it needs to fetch rows corresponding to the bin having  $\mathit{cid}_2$, by generating four trapdoors:  $\mathcal{E}_k(\mathit{cid}_2||1)$, $\mathcal{E}_k(\mathit{cid}_3||1)$, $\mathcal{E}_k(f||1)$, and $\mathcal{E}_k(f||2)$. Finally, based on the retrieve rows, SGX produces the final answer.$\blacksquare$

\section{Data Encryption at Data Provider}
\label{sec:Data encryption and Outsourcing}

\textsc{Concealer} stores data in discredited time slots, called \emph{epochs} or \emph{rounds}. Epoch duration is selected based upon the latency requirements of $\mathcal{SP}$. Executing queries over multiple epochs could lead to inference attacks, and for dealing with it, we will present a method in~ \S\ref{sec:Insert Operation}. This section describes Algorithm~\ref{alg:Data encryption algorithm}, which is executed at $\mathcal{DP}$, for encrypting time-series data (assuming with three attributes location $\mathcal{L}$, time $\mathcal{T}$, and object $\mathcal{O}$) belonging to one epoch. ({{In our experiments~\S\ref{sec:Experimental Evaluation}, we will consider different datasets with multiple columns.}})
Algorithm~\ref{alg:Data encryption algorithm} uses deterministic encryption (DET) to support fast query execution. Since DET produces the same ciphertext for more than one occurrence of the same location and object, to ensure ciphertext indistinguishability, we concatenate each occurrence of the location and observation values with the corresponding timestamp. 

In \textsc{Concealer}, queries retrieve a subset of tuples based on predicates specified over attributes, such as $\mathcal{L}$, $\mathcal{O}$, or both. Queries, further, are always associated with ranges over time (see Table~\ref{tab:Sample queries}). Thus, to support the efficient execution of such queries, \textsc{Concealer} creates a cell-based index over query attributes (\textit{e}.\textit{g}., $\mathcal{L}$ and/or $\mathcal{O}$) along with time. For simplicity, Algorithm~\ref{alg:Data encryption algorithm}  illustrates how a cell-based index is created for  location and time attributes,  $\mathrm{Index}(\mathcal{L},\mathcal{T})$.  Similar indexes can also be created for other attributes, such as $\mathrm{Index}(\mathcal{O},\mathcal{T})$ and $\mathrm{Index}(\mathcal{L},\mathcal{O},\mathcal{T})$. Details of Algorithm~\ref{alg:Data encryption algorithm} is given below:

\medskip
\noindent
\textbf{Key generation (Lines~\ref{ln:key_gen}).} Since using a single key over multiple epochs will result in the identical ciphertext of a value, \textsc{Concealer} produces a key for encryption for each epoch, as $k\leftarrow s_k||\mathit{eid}$, where $s_k$ is the secure key shared between SGX and $\mathcal{DP}$, $\mathit{eid}$ is the epoch-id, which is the starting timestamp of the epoch, and $||$ denotes concatenation. Thus, encrypting a value $v$ using $k$ in two different epochs will produce different ciphertexts. (Only the first $\mathit{eid}$ and epoch duration is provided to SGX to generate other $\mathit{eid}$ to decrypt the data during query execution.)

\medskip
\noindent
\textbf{Tuple encryption (Lines~\ref{ln:function_encrypt_data}-\ref{ln:return_answer}).}
As the tuple arrives, it got appropriately encrypted (Line~\ref{ln:et}) using DET. Note that by encryption over the concatenated time with location and object values, results in a unique value in the entire relation. Now, in order to allocate the cell value to be used as the index, we proceed as follows: Let $|\mathcal{L}|$ be the number of locations and $|\mathcal{T}|$ be the duration of the epoch. \textsc{Concealer} maps the set of location $|\mathcal{L}|$ into a range of values from $1$ to $x \leq |\mathcal{L}|$ using a simple hash function. It, furthermore, partitions $|\mathcal{T}|$ into $y>1$ subintervals of duration, $|\mathcal{T}|/y$, which are then mapped using a hash function into $y > 1$ values. Thus, all tuples of the epoch are distributed randomly over the grid of $x\times y$ (see Example~\ref{sec:Data encryption and Outsourcing} below). Then, $u$ cell-ids ($u<x\times y$) are allocated to grid cells. To refer to the cell-id of a cell, we use the notation $\mathit{cid}_z^{\{p,q\}}$ that shows that the cell $\{p,q\}$ is assigned a cell-id $\mathit{cid}_z$. In \textsc{Step} 3 of query execution~ \S\ref{subsec:Bin-Packing-based BPB Method Query Execution}, it will be clear that we will fetch tuples to answer any query based on cell-ids, instead of directly using query predicates.

Here, we keep two vectors: (\textit{i}) $\mathit{cell\_id}$ of length $x\times y$ to keep the cell-id allocated to each cell of the grid, and (\textit{ii}) $\mathit{c\_tuple}$ of length $u$ to store the number of tuples that have been allocated the same cell-id. During processing a $j^{\mathit{th}}$ tuple, we increment the current counter value of the number of tuples that have the same cell-id by one using $\mathit{c\_tuple}$, and encrypts it. This value will be allocated to $\mathrm{Index}(\mathcal{L}, \mathcal{T})\rangle$ attribute of the $j^{\mathit{th}}$ (Lines~\ref{ln:find_counter}-\ref{ln:encrypt_counter}).

\LinesNotNumbered \begin{algorithm}[!t]
\footnotesize
\textbf{Inputs:} $R$: a relation. $\mathbb{H}$: A hash function. $\mathcal{E}()$: An encryption function. $s_k$: a secret key.

\textbf{Outputs:} $E(R)$: the encrypted relation.

\nl \textbf{Variables:} $\forall c_t \leftarrow 0$, where $1\leq t\leq r$. $x \leftarrow \# \mathbb{H}(\mathit{Domain}(\mathcal{L}))$,
$y \leftarrow \# \mathbb{H}(\mathit{Domain}(\mathcal{T}))$, $\mathit{cell\_id}[x,y]\leftarrow 0$, $\mathit{c\_tuple}[u]\leftarrow 0$. \nllabel{ln:variable_init}


\nl{\bf Function $\boldsymbol{\mathit{key\_gen(s_k)}}$} \nllabel{ln:key_gen}
\Begin{
\nl $k\leftarrow (s_k ||\mathit{eid})$
}
\nl{\bf Function $\boldsymbol{\mathit{encrypt\_data(R)}}$} \nllabel{ln:function_encrypt_data}
\Begin{

\nl \For{$j\in (0,n-1)$}{

\nl $\mathit{Eo}_j \leftarrow \mathcal{E}_k(o_j||t_j)$, \nllabel{ln:eid}
$\mathit{El}_j \leftarrow \mathcal{E}_k(\mathit{l}_j||t_j)$,
$\mathit{Er}_j \leftarrow \mathcal{E}_k(v_j||l_j||t_j)$ \nllabel{ln:et}

\nl {\bf Function $\boldsymbol{\mathit{Cell\textnormal{-}Formation}(j^\mathit{th}\textnormal{  tuple})}$} \nllabel{ln:function_cell_formation}
\Begin{

\nl$p\leftarrow \mathbb{H}(l_j)$, $q\leftarrow \mathbb{H}(t_j)$, $\mathit{cid}_z^{\{p,q\}} \leftarrow \mathit{cell\_id}[p,q]$ \nllabel{ln:find_cell}


\nl$c_t\leftarrow \mathit{c\_tuple}[\mathit{cid}_z^{\{p,q\}}] \leftarrow \mathit{c\_tuple}[\mathit{cid}_z^{\{p,q\}}] + 1$ \nllabel{ln:find_counter}

\nl$\mathit{Ec}_j \leftarrow E_k(\mathit{cid}_z^{\{p,q\}}||c_t)$ \nllabel{ln:encrypt_counter}
}
\nl \Return $E(R)\leftarrow \langle  \mathit{Eo}_j,\mathit{El}_j,\mathit{Er}_j, \mathit{Ec}_j\rangle$ \nllabel{ln:return_answer}
}
}

\nl{\bf Function $\boldsymbol{\mathit{add\_fake\_tuples()}}$} \nllabel{ln:function_add_fake_rows}
\Begin{
\nl \For{$j\in (0,n-1)$}{

\nl Generate fake $\mathit{Eo}_j$, $\mathit{El}_j$, and $\mathit{Er}_j$, and  \nllabel{ln:generate_fake_rows}
$\mathit{Ec}_j \leftarrow \mathcal{E}_k(f||j)$ 

\nl Append the $j^{\mathit{th}}$ fake tuple to the relation $E(R)$ \nllabel{ln:append_rows}
}
}

\nl{\bf Function $\boldsymbol{\mathit{HashChain}(\mathit{c\_tuple}[u],\langle \mathit{Eo},\mathit{El},\mathit{Er}\rangle)}$} \nllabel{ln:function_hash_chain}
\Begin{

\nl \For{$j\in c\_tuple[]$, $\forall p$ tuples with same cell-id}{

\nl
$h_l^j\leftarrow H( El_p)|| (H(El_{p-1}) || (\ldots||(H(El_2)||H(El_1))) \ldots ) ))$

\nl
$h_o^j\leftarrow H( Eo_p)|| (H(Eo_{p-1}) || (\ldots||(H(Eo_2)||H(Eo_1))) \ldots ) ))$

\nl
$h_r^j\leftarrow H( Er_p)|| (H(Er_{p-1}) || (\ldots||(H(Er_2)||H(Er_1))) \ldots )))$

\nl $Ehl^j \leftarrow E(h_l^j)$, $Eho^j \leftarrow E(o_l^j)$, $Ehr^j \leftarrow E(h_r^j)$ \nllabel{ln:hash_function_end}
}}

\nl{\bf Function $\boldsymbol{\mathit{Transmit}( E(R),\mathit{cell\_id}[x,y],\mathit{c\_tuple}[u])}$} \nllabel{ln:function_outsource}
\Begin{

\nl $\mathit{Ecell\_id}[x,y] \leftarrow \mathcal{E}_{\mathit{nd}}(\mathit{cell\_id}[x,y])$,  
 $\mathit{Ec\_tuple}[u] \leftarrow \mathcal{E}_{\mathit{nd}}(\mathit{c\_tuple}[u])$ \nllabel{ln:encrypt_cc}

\nl Permute all the tuples of the encrypted relation $E(R)$

\nl Send $E(R)$, $\mathit{Ecell\_id}[x,y]$, $\mathit{Ec\_tuple}[u]$,
$Ehl^j$, $Eho^j$, $Ehr^j$ \nllabel{ln:outsource_end}
}
\caption{Data encryption algorithm.}
\label{alg:Data encryption algorithm}
\end{algorithm}
\setlength{\textfloatsep}{0pt}

\medskip
\noindent
\textbf{Allocating fake tuples (Lines~\ref{ln:function_add_fake_rows} -\ref{ln:append_rows}).} Since $\mathcal{SP}$ will read the data from DBMS into the enclave, different numbers of rows according to different queries may reveal information about the encrypted data. Thus, to fetch an equal number of rows for any query, $\mathcal{DP}$ needs to share some fake rows. There are two methods for adding the fake rows:

\noindent (\textit{i}) \emph{Equal number of real and fake rows}: This is the simplest method for adding the fake rows. Here, $\mathcal{DP}$ adds ciphertext secure fake tuples. The reason of adding the same number of real and fake rows is dependent on the property of the bin-packing algorithm, which we will explain in~\S\ref{subsec:Bin-Packing-based BPB Method Query Execution} (Theorem~\ref{th:our_bounds}).\footnote{{\scriptsize
For $n$ real tuples, we add a little bit more than $n$ fake tuples in the worst case (Theorem~\ref{th:our_bounds}).}} In $\mathrm{Index}$ attribute, a $j^{\mathit{th}}$ fake tuple contains an encrypted identifier with the tuple-id $j$, denoted by $\mathcal{E}_k(f||j)$, where $f$ is an identifier (known to only $\mathcal{DP}$) to distinguish real and fake tuples.

\noindent (\textit{ii}) \emph{By simulating the bin-creation method}: To reduce the number of fake rows to be sent, we use this method in which $\mathcal{DP}$ simulates the bin-packing algorithm (as will be explained in~\S\ref{subsec:Bin-Packing-based BPB Method Query Execution}) and finds the total number of fake rows required in all bins such that their sizes must be identical. Then, $\mathcal{DP}$ share such ciphertext secure fake tuples with their $\mathrm{Index}$ values, as in the previous method. As will be clear soon by Theorem~\ref{th:our_bounds} in~\S\ref{subsec:bin creation}, in the worst case, both the fake tuple generation methods send the same number of fake tuples, \textit{i}.\textit{e}., an equal number of real and fake tuples.

\medskip
\noindent
\textbf{Hash-chain creations (an optional step) Line~\ref{ln:function_hash_chain}-\ref{ln:hash_function_end}.} $\mathcal{DP}$ creates hash chains over encrypted tuples allocated an identical cell-id, as follows: let $p$ be the numbers of tuples allocated the same cell-ids. Consider $p$ encrypted location ciphertext as: $E(l_1), E(l_2),\ldots, E(l_p)$. Now, $\mathcal{DP}$ executes a hash function as follows:

\centerline{
$h_{l1}\leftarrow H(E(l_1))$}
\centerline{
$h_{l2}\leftarrow H(E(l_2)||h_{l1})$}
\centerline{$\ldots$}
\centerline{
$h_{l}\leftarrow H(E(l_p)||h_{l(p-1)})$}
In the same way, hash digests for other columns are computed, and the final hash digest (\textit{i}.\textit{e}., $h_l$) is encrypted that works as a verifiable tag at $\mathcal{SP}$.

\medskip
\noindent\textbf{Sending data (Line~\ref{ln:function_outsource}-\ref{ln:outsource_end}).} Finally, $\mathcal{DP}$ permutes all encrypted tuples of the epoch to mix fake and real tuples in the relation and sends them with the two encrypted vectors $\mathit{Ecell\_id}[]$ and $\mathit{Ec\_tuple}[]$ and encrypted hash digests.\footnote{{\scriptsize The size of both vectors is significantly smaller (see experimental section~ \S\ref{subsec:Setup and Datasets}).}}


\medskip
\noindent\textbf{Example~\ref{sec:Data encryption and Outsourcing}.}
Now, we explain with the help of an example how encryption algorithm works. Consider six rows of Table~\ref{tab:Cleartext data at the data provider} as the rows of an epoch, and we wish to encrypt those tuples with index on attributes $\mathcal{L}$ and $\mathcal{T}$. Assume that Algorithm~\ref{alg:Data encryption algorithm} creates a $2\times 2$ grid (see Table~\ref{tab:the grid 2times2}) with three cell-ids: $\mathit{cid}_1$, $\mathit{cid}_2$, and $\mathit{cid}_3$. Table~\ref{tab:the grid 2times2} shows two vectors $\mathit{cell\_id}[]$ and $\mathit{c\_tuple}[]$ corresponding to $\mathcal{L}$ and $\mathcal{T}$ attributes. Values in $\mathit{c\_tuple}[]$ show that the number of tuples has been allocated the same cell-id. For instance, $\mathit{c\_tuple}[1]=4$ shows that four tuples are allocated the same cell-id (\textit{i}.\textit{e}., $\mathit{cid}_1$). In Table~\ref{tab:the grid 2times2}, for explanation purposes, we show which rows of Table~\ref{tab:Cleartext data at the data provider} correspond to which cell; however, this information is not stored, only information of vectors $\mathit{cell\_id}[]$ and $\mathit{c\_tuple}[]$ is stored.

The complete output of Algorithm~\ref{alg:Data encryption algorithm} is shown in Table~\ref{tab:Encrypted data at the cloud} for cleartext data shown in Table~\ref{tab:Cleartext data at the data provider}, where $\mathrm{Index}(\mathcal{L},\mathcal{T})$ is the column on which DBMS creates an index. In Table~\ref{tab:Encrypted data at the cloud}, $\mathcal{E}$ refers to DET, $\mathcal{E}_{\mathit{nd}}$ refers to a non-deterministic encryption function, and $k$ be the key used to encrypt the data of the epoch. In addition, we create three hash chains, one hash chain per cell-id. Also, note that this example needs only 2 fake tuples to prevent output-size leakage at $\mathcal{SP}$.$\blacksquare$

\section{Point Query Execution}
\label{sec:point Query Execution}

This section develops a bin-packing-based (BPB) method for executing point queries. Later,~\S\ref{sec:Range Query} will develop a method for range queries. The objectives of BPB method are twofold: first, create identical-size bins to prevent leakages due to output-size, (\textit{i}.\textit{e}., when reading some parts of the data from disk to the enclave), and second, show that the addition of \emph{at most} $n$ fake tuples is enough in the worst case to prevent output-size leakage, where $n$ is the number of real tuples. BPB method partitions the values of $\mathit{c\_tuple}[]$ into almost equal-sized bins, using which a query can be executed. Note that bins are created only once, prior to the first query execution. This section, first, presents the bin-creation method, and then, BPB query execution method.


\subsection{Bin Creation}
\label{subsec:bin creation}
Bins are created inside the enclave using a bin-packing algorithm, after decrypting vector $\mathit{Ec\_tuple}[]$.

\medskip
\noindent\textit{\underline{Bin-packing algorithms.}} A bin-packing algorithm places the given inputs having different sizes to bins of size at least as big as the size of the largest input, without partitioning an input, while tries to use the minimum number of bins. First-Fit Decreasing (FFD) and Best-Fit Decreasing (BFD)~\cite{Coffman:1996:AAB:241938.241940} are the most notable bin-packing algorithms and \emph{ensure that all the bins (except only one bin) are at least half-full.} 


In our context, $u$ cell-ids ($\mathit{cid}_1,\mathit{cid}_2, \ldots, \mathit{cid}_u$) are inputs to a bin-packing algorithm, and the number of tuples having the same cell-id is considered as a weight of the input. Let $\mathit{max}$ be the maximum number of tuples having the same cell-id $\mathit{cid}_i$. Thus, we create bins of size at least $|b|=\mathit{max}$ and execute FFD or BFD over $u$ different cell-ids, resulting in $|\mathit{Bin}|$ bins as an output of the bin-packing algorithm.

\medskip\noindent
\textbf{The minimum number of bins.} Let $n$ be the number of real tuples sent by $\mathcal{DP}$, \textit{i}.\textit{e}., $n=\sum_{i=1}^{i=u}\mathit{c\_tuple}[i]$. Let $|b|$ be the size of each bin. Thus, it is required to divide $n$ inputs into at least $\lceil n/|b|\rceil$ bins.

\begin{theorem}
\label{th:our_bounds}
\textnormal{\textsc{(Upper bounds on the number of bins and fake tuples)}} The above bin-packing method using a bin size $|b|$ achieves the following upper bounds: the number of bins and the number of fake tuples sent by $\mathcal{DP}$ are at most $\frac{2n}{|b|}$ and at most $n+\frac{|b|}{2}$, respectively, where $n\gg |b|$ is the number of real tuples sent by $\mathcal{DP}$.
\end{theorem}

\begin{proof}
A bin $b_i$ can hold inputs whose sum of the sizes is at most $|b|$. Since the FFD or BFD bin-packing algorithm ensures that all the bins (except only one bin) are at least half-full, each bin of size $|b|$ has at least all those cell-ids whose associated number of tuples is at least ${|b|}/{2}$. Thus, all $n$ real tuples can be placed in at most ${n}/({|b|/2})$ bins, each of size $|b|$. Further, since all such ${2n}/{|b|}$ bins are at least half-full, except the last one, we need at most $n+({|b|}/{2})$ more tuples to have all the bins with $|b|$ tuples. Thus, $\mathcal{DP}$ sends at most $n+({|b|}/{2})\approx n$ fake tuples with $n$ real tuples. 
\end{proof}


\medskip\noindent
\textbf{Equi-sized bins.} The output bins of FFD/BFD may have different numbers of tuples. 
Thus, we pad each bin with fake tuples, thereby all bins have $|b|$ tuples.
Let $\mathit{tuple}_{b_i}< |b|$ be the number of tuples assigned to an $i^{\mathit{th}}$ bin (denoted by $b_i$). Here, the ids (\textit{i}.\textit{e}., the value of $\mathrm{Index}$ column) of fake tuples allocated to the bin $b_i$ will be $|b|-\mathit{tuple}_{b_i}$, and all these fake tuple ids cannot be used for padding in any other bin. Thus, for padding, we create disjoint sets of fake tuple ids (see the example below to understand the reason).

\medskip
\noindent
\textbf{Example~\ref{sec:point Query Execution}.1.}
Assume five cell-ids $\mathit{cid}_1,\mathit{cid}_2,\ldots, \mathit{cid}_5$ having the following number of tuples $\mathit{c\_tuple}[5]=\{79,2,73,7,7\}$. 
Here, $\mathit{cid}_1$ has the maximum number of tuples; hence, the bin-size is at least 79. After executing FFD bin-packing algorithm, we obtain three bins, each of size 79: $b_1$: $\langle \mathit{cid}_1\rangle$, $b_2$: $\langle \mathit{cid}_3, \mathit{cid}_2\rangle$, and $b_3$: $\langle \mathit{cid}_5, \mathit{cid}_4\rangle$. Here, bins $b_2$ and $b_3$ needs 4 and 65 fake tuples, respectively. One can think of sending only 65 fake tuples to access bins $b_2$ and $b_3$ to have size 79. However, in the absence of access-pattern hiding techniques, the adversary will observe that any 4 tuples out of 65 fake tuples are accessed in both bins. It will reveal that these four tuples are surely fake, and thus, the adversary may deduce that the bin size of $b_2$ is 75. Thus, $\mathcal{DP}$ needs to send 69 fake tuples in this example.$\blacksquare$



\subsection{Bin-Packing-based (BPB) Query Execution}
\label{subsec:Bin-Packing-based BPB Method Query Execution}
We present BPB method (see pseudocode in Algorithm~\ref{alg:Query execution}) based on the created bins (over location and time attributes). A similar method can be extended for other attributes. BPB method contains the following four steps:

\medskip
\noindent\textbf{\textsc{Step} 0: Bin-creation.} By following FFD or BFD as described above, this step creates bins over cell-ids ($\mathit{c\_tuple}[]$), if bins do not exist.

\medskip
\noindent\textbf{\textsc{Step} 1: Cell identification (Lines~\ref{ln:function_find_cell}-\ref{ln:return_cell_id}).} 
The objective of this step is to find a cell of the grid corresponding to the requested location and time. A query $Q_e=\langle \mathit{qa}, (\mathcal{L}=l,\mathcal{T}=t)\rangle$ is submitted to the enclave that, on the query predicates $l$ and $t$, applies the hash function $\mathbb{H}$, which was also used by $\mathcal{DP}$ (in $\mathit{Cell\textnormal{-}Formation}$ function, Line~\ref{ln:find_cell} of Algorithm~\ref{alg:Data encryption algorithm}). Thus, the enclave knows the cell, say $\{p,q\}$, corresponds to $l$ and $t$. Based on the cell $\{p,q\}$ and using
the vector $\mathit{cell\_id}[]$, it knows the cell-id, say $\mathit{cid}_z$, allocated to the cell $\{p,q\}$.

\medskip
\noindent\textbf{\textsc{Step} 2: Bin identification (Lines~\ref{ln:function_find_bin}-\ref{ln:find_bin_done}).}
Based on the output of \textsc{Step} 1, \textit{i}.\textit{e}., the cell-id $\mathit{cid}_z$, this step finds a bin $b_i$ that contains $\mathit{cid}_z$. Bin $b_i$ may contain several other cell-ids along with identities of the first and the last fake tuples required for $b_i$.

\medskip
\noindent\textbf{\textsc{Step} 3: Query formulation (Lines~\ref{ln:function_generate_queries}-\ref{ln:get_queries}).} After knowing all cell-ids that are required to be fetched for bin $b_i$, the enclave formulates appropriate ciphertexts that are used as queries. Let the set of cell-ids in $b_i$ be $C_{1},C_2,\ldots,C_{\alpha}$, containing $\#_{1}, \#_{2}, \ldots, \#_{\alpha}$ records, respectively. Let the fake tuple range for $b_i$ be $f_l$ and $f_h$ (let $\#_f = f_h - f_l$ be the number of fake tuples that have to be retrieved for $b_i$). The enclave generates $\#_i$ number of queries, as: $\mathcal{E}_k(C_p||j)$, where $1 \leq j \leq \#_i$ for each cell $C_p$ corresponding to $b_i$ and $k$ is the key obtained by concatenating $s_k$ and epoch-id (as mentioned in Line~\ref{ln:key_gen} of Algorithm~\ref{alg:Data encryption algorithm}). Also, it generates $\#_f$ fake queries, one for each of the fake tuples associated with $b_i$.

\noindent
\textit{Advantage of cell-ids.} Now, observe that a bin may contain several locations and time values (or any desired attribute value). Fetching data using cell-id does not need to maintain fine-grain information about the number of tuples per location per time.

\LinesNotNumbered
\begin{algorithm}[!t]

\footnotesize
\textbf{Inputs:}
$\langle\mathit{qa},(\mathcal{L}=l, \mathcal{T}=t)\rangle$: a query $qa$ involving predicates over $\mathcal{L}$ and $\mathcal{T}$ attributes. $\mathit{cell\_id}[x,y]$, $\mathit{c\_tuple}[u]$, $\mathbb{H}$: A hash function, $\mathcal{E}_k()$: An encryption function using a key $k$

$|\mathit{Bin}|$: the number of bins. $\mathit{b}[i][j]$: $i^{\mathit{th}}$ bin having $j$ cell-ids, where $j>0$.

\textbf{Outputs:} A set of ciphertext queries.



\nl{\bf Function $\boldsymbol{\mathit{Query\_Execution(\langle\mathit{qa},(\mathcal{L}=l, \mathcal{T}=t)\rangle)}}$} \nllabel{ln:function_query_execution}
\Begin{

\nl{\bf Function $\boldsymbol{\mathit{Find\_cell(l,t)}}$} \nllabel{ln:function_find_cell}
\Begin{

\nl $p \leftarrow \mathbb{H}(l)$, $q \leftarrow \mathbb{H}(t)$,  
$\mathit{cid}_z \leftarrow \mathit{cell\_id}[p,q]$ \nllabel{ln:find_cell_id_query_exe}

\nl \Return $\mathit{cid}_z$; \textbf{break}  \nllabel{ln:return_cell_id}
}

\nl{\bf Function $\boldsymbol{\mathit{Find\_bin(\mathit{cid}_z, |\mathit{Bin}|, b[\ast][\ast])}}$} \nllabel{ln:function_find_bin}
\Begin{
\nl \For{$\mathit{desired} \in (0,|\mathit{Bin}|-1)$\nllabel{ln:check_each_sbucket}}{\nl \If{$ \mathit{cid}_z\in \mathit{b}[\mathit{desired}][\ast]$}{
\nl \Return $\mathit{b}[\mathit{desired}][\ast]$; \textbf{break} \nllabel{ln:find_bin_done}
}}}

\nl{\bf Function $\boldsymbol{\mathit{Formulate\_queries(\mathit{b}[\mathit{desired}][\ast])}}$} \nllabel{ln:function_generate_queries}
\Begin{

\nl \For{$\forall j\in (\mathit{b}[\mathit{desired}][j])$}{
\nl $\mathit{cell\_id} \leftarrow \mathit{b}[\mathit{desired}][j]$, 
 $\mathit{counter} \leftarrow \mathit{c\_tuple}[\mathit{cell\_id}]$ \nllabel{ln:get_cell_counter}

\nl $\forall \mathit{counter}$, generate ciphertexts $\mathit{E}_k(\mathit{cell\_id}||\mathit{counter})$ \nllabel{ln:get_queries}
}
}
}
\caption{Bin-packing-based query execution method.}
\label{alg:Query execution}
\end{algorithm}
\setlength{\textfloatsep}{0pt}

\medskip
\noindent\textbf{\textsc{Step} 4: Integrity verification and final answers filtering.}
We may optionally verify the integrity of the retrieved tuples. To do so, the enclave, first, creates a hash chain over the real encrypted tuples having the same cell-id, by following the same steps as $\mathcal{DP}$ followed (Lines~\ref{ln:function_hash_chain} of Algorithm~\ref{alg:Data encryption algorithm}). Then, it compares the final hash digest against the decrypted verifiable tag, provided by $\mathcal{DP}$.

Now, to answer the query, the enclave, first, filters those tuples that do not qualify the query predicate, since all tuples of a bin may not correspond to the answer. Thus, decrypting each tuple to check against the query $Q_e=\langle \mathit{qa},(\mathcal{L}=l,\mathcal{T}=t)\rangle$ may increase the computation cost. To do so, after implementing the above-mentioned \textsc{Step} 3, the enclave generates appropriate filter values ($\mathcal{E}_k(l_i||t_i)$ or $\mathcal{E}_k(\mathit{o}_i||t_i)$, which are identical to the created by $\mathcal{DP}$ using Algorithm~\ref{alg:Data encryption algorithm}); while, at the same time, DBMS executes queries on the encrypted data. On receiving encrypted tuples from DBMS, the enclave performs string-matching operations using filters and decrypts only the desired tuples, if necessary. 

\medskip
\noindent\textbf{Example~\ref{sec:point Query Execution}.2.}
Consider the cells created in Example~\ref{sec:Data encryption and Outsourcing}.1, \textit{i}.\textit{e}.,
$\mathit{cell\_id}[]=\mathit{cid}_1,\mathit{cid}_2,\mathit{cid}_1,\mathit{cid}_4$ and $\mathit{c\_tuple}[]=\{4,1,1\}$.
Now, assume that there are two bins, namely $\mathit{b}_1: \langle \mathit{cid}_1\rangle$ and $\mathit{b}_2: \langle \mathit{cid}_2, \mathit{cid}_3 \rangle$.
Consider a query $Q=\langle$count$, (l_2,t_5)\rangle$, \textit{i}.\textit{e}., find the number of people at location $l_2$ at time $t_5$ on the data shown in Table~\ref{tab:Encrypted data at the cloud}. Here, after implementing \textsc{Step} 1 and \textsc{Step} 2 of BPB method, the enclave knows that cell-id $\mathit{cid}_2$ satisfies the query, and hence, the tuples corresponding to bin $b_2$ are required to be fetched. Thus, in \textsc{Step} 3, the enclave generates the following four queries: $\mathcal{E}_k(\mathit{cid}_2||1)$, $\mathcal{E}_k(\mathit{cid}_3||1)$, $\mathcal{E}_k(f||1)$, and $\mathcal{E}_k(f||2)$. Finally, in \textsc{Step} 4, the filtering via string matching is executed over the retrieved four tuples against $\mathcal{E}_k(l_2||t_5)$. Since all the four retrieved tuples have a filter on location and time values, here is no need to decrypt the tuple that does not match the filter $\mathcal{E}_k(l_2||t_5)$.$\blacksquare$




\begin{figure}[!t]
\scriptsize
\begin{mdframed}[style=MyFrame,nobreak=true]			\begin{minipage}{.4\linewidth}
\texttt{max(int x,int y)\{\\
bool getX = ogreator(x, y),\\
return omove(getX, x, y)\\
\}
}
    \subcaption{Oblivious maximum.}
    \label{fig:Oblivious max.}
			\end{minipage}
			\begin{minipage}{.33\linewidth}
\texttt{mov  rcx, x\\
mov  rdx, y\\
cmp  rcx, rdx\\
setg al\\
retn\\
}
    \subcaption{Oblivious compare: \texttt{ogreator}.}
    \label{fig:ogreator}
			\end{minipage}			
			\begin{minipage}{.25\linewidth}
\texttt{mov  rcx, cond\\
mov  rdx, x\\
mov  rax, y\\
test rcx, rcx\\
cmovz rax, rdx\\
retn
}
\subcaption{Oblivious move: \texttt{omove}.}
\label{Oblivious move.}
\end{minipage}
\end{mdframed}
\caption{Register-oblivious operators~\cite{DBLP:conf/uss/OhrimenkoSFMNVC16}.}
\label{fig:Register-oblivious operators}
\end{figure}

\subsection{Oblivious Trapdoor Creation \& Filtering Steps}
\label{subsec:Oblivious Trapdoor Creation and Filtering Steps}
Steps for generating trapdoor (\textsc{Step} 3) and answer filtering (\textsc{Step} 4), in~\S\ref{subsec:Bin-Packing-based BPB Method Query Execution}, were not oblivious due to side-channel attacks (\textit{i}.\textit{e}., access-patterns revealed via cache-lines and branching operations) on the enclave. Thus, we describe how can the enclave produce queries and process final answers obliviously for preventing side-channel attacks.

\medskip
\noindent{\textbf{\textsc{Step} 3.}}
Let $\#_{Cmax}$ be the maximum cells required to form a bin. Let $\#_{max}$ be the maximum tuples with a cell-id. Let $\#_i$ be the number of tuples with a cell-id $C_i$. For a bin $b_i$, the enclave generates $\#_{Cmax}\times \#_{max}$ numbers of queries: $\mathcal{E}_k(C_i||j,v)$, where $1 \leq j \leq \#_{max}$, $C_i\leq\#_{Cmax}$, and $v=1$ if $j\leq \#_i$ and $C_i$ is required for $b_i$; otherwise, $v=0$. Note that \emph{this step produces the same number of queries for each cell and each bin}.

Let $\#_{\mathit{fmax}}$ be the maximum fake tuples required for a bin.
Let $\#_{\mathit{fb}_i}$ be the maximum fake tuples required for $b_i$. The enclave generates $\#_{\mathit{fmax}}$ number of fake queries: $\mathcal{E}_k(f||j,v)$, where $1 \leq j \leq \#_{\mathit{fmax}}$ and $v=1$ if $j\leq \#_{\mathit{fb}_i}$; otherwise, 0. This step \emph{produces the same number of fake queries for any bin}. Finally, the enclave sorts all real and fake queries based on value $v$ using a data-independent sorting algorithm (\textit{e}.\textit{g}., bitonic sort~\cite{bitonicsort}), such that all queries with $v=1$ precede other queries, and sends only queries with $v=1$ to the DBMS.

\medskip
\noindent{\textbf{\textsc{Step} 4.}} The enclave reads all retrieved tuples and appends $v=1$ to each tuple if they satisfy the query/filter; otherwise, $v=0$. Particularly, an $i^{\mathit{th}}$ tuple is checked against each filter, and once it matches one of the filters, $v=1$ remains unchanged; while the value of $v=1$ is overwritten for remaining filters checking on the $i^{\mathit{th}}$ tuple. It hides that which filter has matched against a tuple. Then, based on $v$-value, it sorts all tuples using a data-independent algorithm.\footnote{{\scriptsize If all tuples can reside in the enclave, then bitonic sort is enough. Otherwise, to obliviously sort the tuples, we use column sort~\cite{columnsort} instead of the standard external merge sort.}} (After this all tuples with $v=1$ are decrypted and checked against the query by following the same procedure, if needed.)

\noindent
{\textit{Branch-oblivious computation}.} Note that after either generating an equal number of queries for any bin or filtering the retrieved tuples using a data-oblivious sort, the entire computation is still vulnerable to an adversary that can observe conditional branches, \textit{i}.\textit{e}., an if-else statement used in the comparison. Thus, to overcome such an attack, we use the idea proposed by~\cite{DBLP:conf/uss/OhrimenkoSFMNVC16}.~\cite{DBLP:conf/uss/OhrimenkoSFMNVC16} suggested that any computation on registers cannot be observed by an adversary since register contents are not accessible to any code outside of the enclave; thus, register-to-register computation is oblivious. For this,~\cite{DBLP:conf/uss/OhrimenkoSFMNVC16} proposed two operators: \texttt{omove} and \texttt{ogreater}, as shown in Figure~\ref{fig:Register-oblivious operators}. For any comparison in the enclave, we use these two operators. Readers may find additional details in~\cite{DBLP:conf/uss/OhrimenkoSFMNVC16}.




\bgroup
\def\arraystretch{1.5}
\begin{table}[!t]
\centering 
\begin{tabular}{l||l|l|l|l|}\hline

  $T_4$ &$\mathit{cid}_1^{\{1,1\}}=40$ & $\mathit{cid}_6^{\{1,2\}}=30$ & $\mathit{cid}_7^{\{1,2\}}=2$ & $\mathit{cid}_{11}^{\{1,4\}}=9$ \\\hline

  $T_3$ & $\mathit{cid}_2^{\{2,1\}}=50$ & $\mathit{cid}_7^{\{2,2\}}=50$ & $\mathit{cid}_6^{\{2,3\}}=21$ & $\mathit{cid}_9^{\{2,4\}}=60$ \\\hline

  $T_2$ & $\mathit{cid}_3^{\{3,1\}}=60$ & $\mathit{cid}_{11}^{\{3,2\}}=40$ & $\mathit{cid}_4^{\{3,3\}}=45$ & $\mathit{cid}_8^{\{3,4\}}=48$ \\\hline

  $T_1$ & $\mathit{cid}_3^{\{4,1\}}=40$ & $\mathit{cid}_{10}^{\{4,2\}}=50$ & $\mathit{cid}_{10}^{\{4,3\}}=10$ & $\mathit{cid}_5^{\{4,4\}}=50$ \\\hline\hline

  ~& $l_1$ & $l_2$ & $l_3$ & $l_4$
\end{tabular}
\caption{A $4\times 4$ grid.}
\label{tab:4by4 grid}
\BBB
\end{table}
\egroup

\section{Range Query Execution}
\label{sec:Range Query}

This section develops an algorithm for executing range queries, by modifying BPB method, given in~ \S\ref{subsec:Bin-Packing-based BPB Method Query Execution}. For simplicity, we consider a range condition on time attribute. 
For illustration purposes, this section uses a $4\times 4$ grid (see Table~\ref{tab:4by4 grid}, which was created by $\mathcal{DP}$ using Algorithm~\ref{alg:Data encryption algorithm},~ \S\ref{sec:Data encryption and Outsourcing}) corresponding to location and time attributes of a relation. In this grid, 11 cell-ids are used, and a number in a cell shows the number of tuples allocated to the cell. The notation $T_i$ shows an $i^{\mathit{th}}$ sub-time interval (after creating a grid using Algorithm~\ref{alg:Data encryption algorithm}~ \S\ref{sec:Data encryption and Outsourcing}). 

\subsection{Trivial Solution: Converting a Range Query into Many Point Queries} Recall that the data provider outsources two vectors $\mathit{cell\_id}[]$ and $\mathit{c\_tuple}[]$. Based on these vectors, a trivial method depending on the bin-packing-based method can be developed to answer a range query, as follows:
\begin{enumerate}[leftmargin=0.2in]
\item
Find all cells and their cell-ids that cover the requested range (by implementing \textsc{Step} 1 of BPB method~\S\ref{subsec:Bin-Packing-based BPB Method Query Execution}).
  \item
Find all those bins that cover the desired cells-ids (by implementing \textsc{Step} 2 of BPB method~\S\ref{subsec:Bin-Packing-based BPB Method Query Execution}).
  \item
Fetch all tuples corresponding to the bins by forming appropriate queries (using \textsc{Step} 3 of BPB method~\S\ref{subsec:Bin-Packing-based BPB Method Query Execution}) and process all retrieved tuples inside the secure hardware (using \textsc{Step} 4 of BPB method~\S\ref{subsec:Bin-Packing-based BPB Method Query Execution}).
\end{enumerate}
But, converting a range query into many point queries may not be efficient, in terms of the number of tuples to be fetched from the cloud and/or the number of tuples to be processed by the secure hardware; see the following example.

\noindent
\textbf{Example~\ref{sec:Range Query}.1.} After implementing the bin-packing-algorithm on the cell-ids $\mathit{c\_tuple}[11]$ (see Table~\ref{tab:4by4 grid}), we may obtain the following six bins:
$b_1: \langle \mathit{cid}_3\rangle$,
$b_2: \langle \mathit{cid}_9,\mathit{cid}_1 \rangle$,
$b_3: \langle \mathit{cid}_{10},\mathit{cid}_4 \rangle$,
$b_4: \langle \mathit{cid}_7,\mathit{cid}_8 \rangle$,
$b_5: \langle \mathit{cid}_6,\mathit{cid}_{11} \rangle$, and
$b_6: \langle \mathit{cid}_2,\mathit{cid}_5\rangle$. Consider a query to count the number of tuples at the location $l_1$ during a given time interval that is covered by $T_2$ to $T_4$. This query is satisfied by cells-ids $\mathit{cid}_{1}$, $\mathit{cid}_{2}$, and $\mathit{cid}_{3}$.
The cell-ids $\mathit{cid}_{1}$, $\mathit{cid}_{2}$, and $\mathit{cid}_{3}$ belong to bins $b_2$, $b_6$, and $b_1$, respectively. Thus, we need to fetch tuples corresponding to three bins: $b_1$, $b_2$, and $b_6$, \textit{i}.\textit{e}., fetching 300 tuples from the cloud, while only 150 tuples satisfy the query.


\subsection{Enhanced Bin-Packing-Based (eBPB) Method}
\label{subsec:Enhanced Bin-Packing-Based eBPB Method}
eBPB method requires $\mathcal{DP}$ to send the number of tuples in each cell of the grid with the vector $\mathit{cell\_id}[]$. Thus, it avoids sending the vector $\mathit{c\_tuple}[]$. For example, for the grid shown in Table~\ref{tab:4by4 grid}, 
$\mathit{cell\_id}[4,4]=\{
(1,40),(6,30),(7,2),$ $(11,9),(2,50),(7,50),(6,21),(9,60),(3,60),(11, 40),(4,45),$ $(8,48),(3,40),(10,50), (10,10), (5,50) \}$.
%
This information helps us in creating bins more efficiently for a range query, as follows:

\medskip
\noindent\textbf{\textsc{Step} 1: Preliminary step.}
The enclave decrypts the encrypted vector $\mathit{Ecell\_id}[]$.


\medskip
\noindent\textbf{\textsc{Step} 2: Finding top-$\ell$ cell-ids.}
Find top-$\ell$ cells having the maximum number of tuples in one of the locations, where $\ell$ is the number of cells required to answer the range query. Say, location $l_i$ has top-$\ell$ cells that have the maximum number of tuples, denoted by $\mathit{bsize}$ tuples.



\medskip\noindent\textbf{\textsc{Step} 3: Create bins.} Execute this step either if $\ell$ cells required for the current query are more than the cells required for any previously executed query or it is the first query. Fix the bin size to $\mathit{bsize}$ and execute FFD that takes $\mathit{cid}_z^{\{p,q\}}$ as inputs and the number of tuples having $\mathit{cid}_z^{\{p,q\}}$ as the weight of the input. If the bin does not have $\mathit{bsize}$ number of tuples, add fake tuples to the bin. It results in $|Bin|$ number of bins and then, use all such bins for answering any range covered by $\ell$ cells.


\medskip\noindent\textbf{\textsc{Step} 4: Query formulation and final answers filtering.} Find the desired bin satisfying the range query and formulate appropriate queries, as we formed in \textsc{Step} 3 of BPB method~\S\ref{subsec:Bin-Packing-based BPB Method Query Execution}. The DBMS executes all queries and provides the desired tuples to the enclave. The enclave executes the final processing of the query, likewise \textsc{Step} 4 of BPB method~\S\ref{subsec:Bin-Packing-based BPB Method Query Execution}. \textit{Note that for oblivious query formulation and result filtering, we use the same method as described in~\S\ref{subsec:Oblivious Trapdoor Creation and Filtering Steps}.}

\medskip\noindent
\textbf{Example~\ref{subsec:Enhanced Bin-Packing-Based eBPB Method}.1.}
 Consider a query to count the number of tuples at the location $l_1$ during a given time interval that is covered by $T_2$ to $T_4$. This query spans over three cells; see Table~\ref{tab:4by4 grid}.
Here, the maximum number of tuples in any three cells at locations $l_1$, $l_2$, $l_3$, and $l_4$ are $60+50+40=150$, $50+50+40=140$, $45+21+5=71$, and $60+50+48=158$, respectively. Thus, the bin of size 158 can satisfy any query that spans over any three cells (arranged in a column) of the grid. $\blacksquare$


\medskip\noindent\textbf{Example~\ref{subsec:Enhanced Bin-Packing-Based eBPB Method}.2:} \textit{attack on eBPB}.
Consider the following queries on data shown in Table~\ref{tab:4by4 grid}:
(\textit{$Q_1$}) retrieve the number of tuples having location $l_1$ during
a given time interval that is covered by $T_1$ and $T_2$, and
(\textit{$Q_2$}) retrieve the number of tuples having location $l_1$ during
a given time interval that is covered by $T_2$ and $T_3$.
Answering $Q_1$ and $Q_2$ may reveal the number of tuples having $T_1$, $T_2$, and $T_3$, as: in answering \textit{$Q_2$} we do not retrieve 40 tuples (corresponding to $\{4,1\}$ cell) that were sent in answering \textit{$Q_1$} and retrieve 50 new tuples (corresponding to $\{2,1\}$ cell). It, also, reveals that 60 tuples (corresponding to $\{3,1\}$ cell) belong to $T_2$. Note that all such information was not revealed, prior to query execution, due to the ciphertext indistinguishable dataset.$\blacksquare$ 

\subsection{Highly Secured Range Query--- winSecRange}
\label{subsec:Optimizing Range Query}
We, briefly, explain a method to prevent the above-mentioned attacks on a range query. Particularly, we fix the length of a range, say $\lambda>1$, and discretize $n$ domain values, say $v_1, v_2, \ldots, v_n$, into $\lceil \frac{n}{\lambda} \rceil$ intervals (denoted by $\mathcal{I}_i$, $1\leq i \leq \lceil \frac{n}{\lambda}\rceil$), as:
$\mathcal{I}_1=\{v_1, v_2, \ldots, v_{\lambda}\}$,
$\ldots$,
$\mathcal{I}_{\lceil \frac{n}{\lambda}\rceil}=\{v_{n-\lambda}, \ldots, v_{n-1},v_n\}$. Here, the bin size equals to the maximum number of tuples belonging to an interval, and bins are created for each interval \emph{only once}. For example, consider 12 domain values: $v_1, v_2, \ldots v_{12}$, and $\lambda= 3$. Thus, we obtain intervals:
$\mathcal{I}_1=\{v_1, v_2, v_3\}$,
$\mathcal{I}_2=\{v_4, v_5, v_6\}$,
$\mathcal{I}_3=\{v_7, v_8, v_9\}$, and
$\mathcal{I}_4=\{v_{10}, v_{11}, v_{12}\}$. Here, four bins are created, each of size equals to the maximum number of tuples in any of the intervals. Now, we can answer a range query of length $\beta$ by using one of the following methods:
\begin{enumerate}[leftmargin=0.2in]
    \item 
    $\beta\leq \lambda$ and $\beta\in \mathcal{I}_i$: Here, the entire range exists in $\mathcal{I}_i$; hence, we retrieve only a single entire bin satisfying the range. E.g., if a range is $[v_1, v_2]$, then we need to retrieve the bin corresponding to $\mathcal{I}_1$.
\item 
$\beta\leq \lambda$ and $\beta\in \{\mathcal{I}_i,\mathcal{I}_j\}$:
It may be possible that $\beta\leq \lambda$ but the range $\beta$ lies in $\mathcal{I}_i$ and $\mathcal{I}_j$, $i\neq j$. Thus, we need to retrieve two bins that cover $\mathcal{I}_i$ and $\mathcal{I}_j$, and hence, we also prevent the attack due to sliding the time window (see Example~\ref{subsec:Enhanced Bin-Packing-Based eBPB Method}.2). For example, if a range is $[v_2, v_4]$, then we need to retrieve the bins corresponding to $\mathcal{I}_1$ and $\mathcal{I}_2$.
\item 
$\beta = z \times \lambda$: Here, a range may belong to at most $z+2$ intervals. Thus, we may fetch at most $z+1$ bins satisfying the query. E.g., if a range is $[v_3, v_8]$, then this range is satisfied by intervals $\mathcal{I}_1$, $\mathcal{I}_2$, and $\mathcal{I}_3$; thus, we fetch bins corresponding to $\mathcal{I}_1$, $\mathcal{I}_2$, and $\mathcal{I}_3$.

\end{enumerate}

\section{Supporting Dynamic Insertion}
\label{sec:Insert Operation}

Dynamic insertion in \textsc{Concealer} is supported by batching updates into rounds/epochs, similar to~\cite{DBLP:conf/sigmod/DemertzisPPDG16}. Tuples inserted in an $i^{\mathit{th}}$ period are said to belong to the round $\mathit{rd}_i$ or epoch $\mathit{eid}_i$. The insertion algorithm is straightforward. \textsc{Concealer} applies Algorithm~\ref{alg:Data encryption algorithm} on tuples of epochs prior to sending them to $\mathcal{SP}$. Note that since Algorithm~\ref{alg:Data encryption algorithm} for distinct rounds is executed independently, the tuples corresponding to the given attribute value (\textit{e}.\textit{g}., location-id) may be associated with different bins in different rounds. Retrieving tuples of a given attribute value across different rounds needs to be done carefully, since it might result in leakage, as shown next.


\medskip
\noindent\textbf{Example~\ref{sec:Insert Operation}.1.} 
Consider that the bin size is three, and we have the following four bins for each round of data insertion, where a bin $b_i$ stores tuples of a location $l_j$:

\centerline{
$\mathit{rd}_1: \quad b_1: \langle l_1, l_2, l_3\rangle \quad b_2: \langle l_4, l_4, l_4\rangle \quad b_3: \langle l_5, l_5, l_5\rangle \quad b_4: \langle l_6, l_6, l_6\rangle$}
\centerline{$\mathit{rd}_2: \quad b_1^{\prime}: \langle l_1, l_1, l_1\rangle \quad b_2^{\prime}: \langle l_2, l_2, l_2\rangle \quad b_3^{\prime}: \langle l_3, l_3, l_3\rangle \quad b_4^{\prime}: \langle l_4, l_5, l_6\rangle$}
Now, answering a query for $l_1$ fetches bins $b_1$ and $b_1^{\prime}$; a query for $l_2$ fetches $b_1$ and $b_2^{\prime}$; and a query for $l_3$ fetches $b_1$ and $b_3^{\prime}$. Note that here $b_1$ is retrieved with three new bins ($b_1^{\prime}$, $b_2^{\prime}$, $b_3^{\prime}$); it reveals that $b_1$ has three distinct locations. Similarly, $b_4^{\prime}$ will be retrieved with three older bins ($b_2$, $b_3$, and $b_4$). Thus, the query execution on older and newer data reveals additional information to the adversary. $\blacksquare$

To prevent such attacks, we need to appropriately modify our query execution methods. In our technique, we will assume that bins across all rounds are of a fixed size, $|b|$,\footnote{{\scriptsize We are not interested in hiding different numbers of tuples in different rounds, but using fake tuples it can be prevented, if desired.}} and the number of tuples for a given attribute value (\textit{i}.\textit{e}., location) fits within a bin (\textit{i}.\textit{e}., $\leq |b|$). Our idea is inspired by Path-ORAM~\cite{DBLP:journals/jacm/StefanovDSCFRYD18}, while we overcome the limitation of Path-ORAM that achieves indistinguishability for query execution by keeping a meta-index structure at the trust entity. Note that Path-ORAM builds a binary tree index on the records. To retrieve a single record, Path-ORAM fetches $\BigO(\log n)$ records and rewrites them under a different encryption. Since Path-ORAM uses an external data structure, it cannot be used for our purpose as argued in~\S\ref{sec:Introduction}. Below, we provide our modified query execution strategy.



\medskip
\noindent \emph{\textbf{Executing queries.}}
Let $rd_i$, $rd_j$, $rd_k$, and $rd_l$ be four consecutive rounds of data insertion. Let $q$ be a query that spans over $rd_j$, $rd_k$, and $rd_l$ rounds; however, only rounds $rd_j$ and $rd_l$ have bins that satisfy query $q$. For answering $q$, the modified query execution method takes the following three steps:
(\textit{i}) The enclave fetches the desired bin from $rd_j$ and $rd_l$ rounds by following methods given in~\S\ref{subsec:Bin-Packing-based BPB Method Query Execution} and~\S\ref{sec:Range Query}, with randomly selected $\log |\mathit{Bin}|-1$ additional bins from each $rd_j$ and $rd_l$ round, where $|\mathit{Bin}|$ are created for each round using Algorithm~\ref{alg:Query execution}.
(\textit{ii}) The enclave fetches $\log |\mathit{Bin}|$ bins from round $rd_k$, to hide the fact that $rd_k$ does not satisfy the query.
(\textit{iii}) For round $rd_x$, $x\in \{j,k,l\}$, the enclave, first, permutes the retrieved data of $rd_x$ and encrypts with a new key.\footnote{\scriptsize The key $k$ for re-encryption is generated as: $k\leftarrow s_k||\mathit{eid}||\mathit{counter}$, where SGX maintains a counter for each round,and increment it by one whenever the data of a round  is read in SGX and rewritten.} The newly encrypted data replaces the older data of $rd_x$. 

\noindent\textit{Aside.} Since we rewrite tuples of retrieved bins, when asking a query for another value belonging to the previously fetched bin (\textit{e}.\textit{g}., query for $l_2$ in Example~\ref{sec:Insert Operation}.1), the adversary cannot link bins of different rounds of data insertion based on attribute values in the bins.

\section{Security Properties}
\label{sec:Security Threats and Properties}

This section presents the desired security requirements, discusses which requirements are satisfied by \textsc{Concealer}, and information leakages from the algorithms. To develop applications on top of spatial time-series dataset at an untrusted $\mathcal{SP}$, a system needs to satisfy the following security properties:

\medskip\noindent
\textbf{Ciphertext indistinguishability:} property requires that any two or more occurrences of a cleartext value look different in the ciphertext. Thus, by observing the ciphertext, an adversary cannot learn anything about encrypted data. 
\textsc{Concealer} satisfies this property by producing unique ciphertext for each tuple by Algorithm~\ref{alg:Data encryption algorithm} (Line~\ref{ln:function_cell_formation}).

\medskip\noindent
\textbf{Data integrity:} property requires that if the adversary injects any false data into the real dataset, it must be detected by a trusted entity. \textsc{Concealer} ensures integrity property by maintaining hash chains over the encrypted tuples and sharing encrypted verifiable tags, which helps SGX to detect any inconsistency between the actual data shared by $\mathcal{DP}$ and the data SGX accesses from the disk at $\mathcal{SP}$.

\medskip\noindent
\textbf{Query execution security:} requires satisfying output-size prevention, indistinguishability under chosen keyword attacks (IND-CKA), and forward privacy.

\medskip\noindent
\textbf{\textit{Output-size prevention:}} property requires that the number of tuples corresponding to a value, \textit{e}.\textit{g}., $\mathcal{L}$, (or a value corresponding to a combination of attributes, \textit{e}.\textit{g}., $\mathcal{L}$ and $\mathcal{T}$) \textit{i}.\textit{e}., the volume of the value, is not revealed, and only the maximum output-size/volume of the value is revealed. \textsc{Concealer} ensures this property by retrieving a fixed-size bin from DBMS into SGX, regardless of the query predicates.

\medskip\noindent
\textbf{\textit{IND-CKA.}}  
IND-CKA~\cite{DBLP:journals/jcs/CurtmolaGKO11} prevents leakages other than what an adversary can gain through information about
(\textit{i}) \emph{metadata}, \textit{i}.\textit{e}., size of database/index, known as \emph{setup leakage} $\mathfrak{L}_s$ in~\cite{DBLP:journals/jcs/CurtmolaGKO11}, and
(\textit{ii}) \emph{query execution} that results in \emph{query leakage} $\mathfrak{L}_q$ in~\cite{DBLP:journals/jcs/CurtmolaGKO11} and includes search-patterns and access-patterns. Again, note that by revealing the access-patterns, IND-CKA is prone to attacks based on the output-size.

While \textsc{Concealer} leaks $\mathfrak{L}_s$ by the size of database/index and $\mathfrak{L}_q$ by fetching data in the form of a bin, it does not reveal information based on the output-size, except a constant output-size for all query predicates. (Also, since by fetching a bin, it does not reveal which rows of the bin satisfy the answer, it hides partial access-patterns.) Thus, \textsc{Concealer} improves the security guarantees of IND-CKA.


\medskip\noindent
\textit{\textbf{Forward privacy:}} property requires that newly inserted tuples cannot be linked to previous search queries, \textit{i}.\textit{e}., the adversary that have collected trapdoors for previous queries, cannot use them to match newly added tuples. \textsc{Concealer} ensures forward privacy by, first, producing a different ciphertext of an identical value over two different epochs using two different keys (as mentioned in~\S\ref{sec:Data encryption and Outsourcing}), and then, re-encrypting the tuples of different epochs using different keys during query execution spanning over multiple epochs (as mentioned in~\S\ref{sec:Insert Operation}).


\medskip\medskip
Now, we discuss leakages from different query execution methods, as follows:

\medskip
\noindent\textbf{BPB information leakage discussion.}
BPB method prevents the attacks based on output-size by fetching an identical number of tuples for answering any query. It reveals the dataset and index sizes stored in DBMS  (as $\mathfrak{L}_s$  condition of IND-CKA~\cite{DBLP:journals/jcs/CurtmolaGKO11}). BPB method, also, reveals \emph{partial} access- and search-patterns, which means that for a group of queries it reveals a fixed bin of tuples, and thus, hides which of the tuples of bin satisfy a particular query ($\mathfrak{L}_q$ conditions of IND-CKA). Recall that an index, \textit{e}.\textit{g}., B-tree index, on the desired attribute is created by the underlying DBMS. To show that the index will not lead to additional leakages other than $\mathfrak{L}_s$ and $\mathfrak{L}_q$, we follow the identical strategy to prove a technique is IND-CKA secure or not. In short, we need to show that a simulator not having the original data can also produce the index attribute based on $\mathfrak{L}=\{\mathfrak{L}_s$, $\mathfrak{L}_q\}$, \textit{i}.\textit{e}., BPB method is secure if a ``fake'' attribute can mimic the real index attribute, (and hence, mimic the real index). Note that like SSEs, the simulator having only $\mathfrak{L}$ can generate a fake dataset, and hence, the index attribute can mimic the real index attribute; thus, the adversary cannot deduce additional information based on $\mathfrak{L}$.

Also, note that in oblivious \textsc{Step} 3, the enclave generates the same number of real/fake queries regardless of a bin and sorts them using a data-independent algorithm, which hides access-patterns in SGX. Also, it processes all retrieved tuples against the query and does oblivious sorting in \textsc{Step} 4. Thus, it also does not reveal access-patterns (by missing any tuple to process).


\medskip\noindent\textbf{eBPB and winSecRange information leakage discussion.}
 eBPB method incurs leakages $\mathfrak{L}_s$ and $\mathfrak{L}_q$. Based on $\mathfrak{L}_q$, we may reveal that a range query is spanning over at most $\ell$ cells. Hence, \emph{it may also reveal the exact data distribution}, by fetching the same real tuple multiple times for multiple range queries, which we illustrated in Example~\ref{subsec:Enhanced Bin-Packing-Based eBPB Method}.2. To overcome such information leakage, winSecRange fetches a fixed size interval, regardless of the query range. Thus, while winSecRange reveals $\mathfrak{L}_s$ and $\mathfrak{L}_q$, it does not reveal any information based on the output-size. 

\bgroup
\def\arraystretch{1.3}
\begin{table*}[!t]
\BBB
  \centering

  \footnotesize
  \begin{tabular}{|p{5.9cm}|p{10.5cm}|}
  \hline
  \textbf{Queries} & \textbf{Execution (filtering, decryption, and final processing) by the secure hardware} \\\hline\hline
  \multicolumn{2}{|c|}{Location and time attributes} \\\hline

  Q1: \# observations at $l_i$ during time $t_1$ to $t_x$ & SM using the filters $El_1\leftarrow E_k(l_i||t_1)$, $El_2\leftarrow \mathcal{E}_k(l_i||t_2)$, $\ldots$,
  $El_k\leftarrow \mathcal{E}_k(l_i||t_x)$. No decryption needed.\\\hline

  Q2: Locations that have top-k observations during $t_1$ to $t_x$   & \multirow{2}{*}{\parbox{10cm}{SM using the filters $El_m \leftarrow \mathcal{E}_k(l_i||t_j)$ where $i\in \mathit{Domain}(\mathcal{L})$  and $j\in \{t_1, t_x\}$, and then, decrypt $\mathcal{E}_k(l||t||o)$ of  qualified tuples only for final processing.}} \\
  \hhline{-~} Q3: Locations that have at least 10 observations during $t_1$ to $t_x$ & \\\hline\hline

  \multicolumn{2}{|c|}{Observation and time attributes} \\\hline
  Q4: Which locations have an observation $o_i$ during $t_1$ to $t_x$ & SM using $\mathit{Eo}_j \leftarrow \mathcal{E}_k(o_i||t_j)$, $j\in \{t_1, t_x\}$, and then, decrypt $\mathcal{E}_k(l||o||t)$ of qualified tuples to know locations. \\\hline\hline

  \multicolumn{2}{|c|}{Observation, location, and time attributes} \\\hline
  Q5: \# an observation $o_i$ has happened at $l_i$ during $t_1$ to $t_x$ & SM using $\mathit{Eo}_j \leftarrow \mathcal{E}_k(o_i||t_j||l_i)$, where $j\in \{t_1, t_x\}$. No decryption needed. \\\hline

\end{tabular}
\caption{Sample queries. \emph{Notation}: SM: String matching.}
\label{tab:Sample queries}
\BBB\BBB\BBB
\end{table*}
\egroup

\medskip\noindent\textbf{Insert operation information leakage discussion.} Our insert operation satisfies forward privacy property. Since for encrypting tuples of an epoch, we generate a key that is unique among all keys generated for any epoch. Thus, based on the previous query trapdoors, the adversary cannot use them to link new tuples. Furthermore, our insert operation hides the distribution leakage due to executing queries over multiple epochs, since we fetch additional tuples from each epoch that lies in the query range and re-write all tuples.

\section{Preventing Attacks due to Query Workload}
\label{sec:Preventing Attacks due to Query Workload}
The (identical-sized) bins formed over different cell-ids may contain different numbers of cell-ids, and that indicates bins may have different numbers of attribute values. Thus, retrieval of bins may reveal data distribution, when executing multiple queries (under the assumption of uniform query workload, \textit{i}.\textit{e}., a query may arrive for each domain value).

\noindent\textbf{Example~\ref{sec:Preventing Attacks due to Query Workload}.1.} Consider 12 bins ($b_1, b_2, \ldots, b_{12}$) of an identical size having the following number of unique values: 1, 2, 9, 1, 2, 10, 1, 1, 1, 8, 2, 7, respectively. Here, if the query workload is uniformly distributed, then all bins will be retrieved different numbers of times. For instance, bins $b_1$, $b_4$, $b_7$, $b_8$, $b_9$ will be retrieved only once, while bin $b_6$ will be retrieved 10 times. It may reveal that bin $b_6$ may have 10 unique values, while bins $b_1$, $b_4$, $b_7$, $b_8$, $b_9$ have only a single value in each. Thus, having equal-sized bins does not help us in preventing data distribution due to query execution.

To prevent such an attack, our objective is to retrieve all the bins almost an equal number of times, under the assumption of uniform query workload. To achieve this, we create super-bins over the bins created in~\S\ref{subsec:Bin-Packing-based BPB Method Query Execution} or in~\S\ref{sec:Range Query}, as follows:
\begin{enumerate}[leftmargin=0.2in]
    \item 
    Sort all the bins in decreasing order of the number of unique values in a bin. Note that all bins have an equal number of rows, but an unequal number of unique values.
\item
Select $f>1$ such that $f$ divides the number of bins evenly and create $f$ super-bin.
\item
Select $f$ largest bins from the sorted order and allocate one in each super-bin. At the end of this step, all $f$ super-bin must have one bin.
\item
In the $i^{\mathit{th}}>1$ iteration, select the next bin, say $b_j$ from the sorted order and find a super-bin, say $S_k$, that is containing $i-1$ number of bins and have the fewest number of unique values in all the allocated bins to the super-bin $S_k$. Allocate the bin $b_j$ to the super-bin $S_k$. Otherwise, select a different super-bin.

\end{enumerate}

Now, for answering a query, the secure hardware will execute the steps given in~\S\ref{subsec:Bin-Packing-based BPB Method Query Execution} (to create bins to answer point queries) or the steps given in~\S\ref{sec:Range Query} (to create bins to answer range queries). After that to avoid workload attack, the secure hardware executes the above-mentioned steps on the created bins. For example, to avoid the workload attack on 12 bins given in Example~\ref{sec:Preventing Attacks due to Query Workload}.1, the following four super-bins may be created:
$S_1$: $\langle b_6,b_7,b_4 \rangle$,
$S_2$: $\langle b_3,b_5,b_8\rangle$,
$S_3$: $\langle b_{10},b_2,b_9\rangle$,
$S_4$: $\langle b_{12},b_{11},b_1\rangle$.
Now, observe that under the assumption of uniform query distribution, $S_1$, $S_2$, $S_3$, $S_4$ will be fetched 12, 12, 11, 10 times, respectively.

\section{Experimental Evaluation}
\label{sec:Experimental Evaluation}

This section shows the experimental results of \textsc{Concealer} under various settings and compares them against prior cryptographic approaches.

\subsection{Datasets, Queries, and Setup}
\label{subsec:Setup and Datasets}

\medskip\noindent\textbf{Dataset 1: Spatial time-series data generation.} To get a real spatial time-series dataset, we took our organization WiFi user connectivity dataset over 202 days having 136M(illion) rows. 
The IT department manages more than 2000 WiFi access-points (AP) by which they collect tuples of the form $\langle l_i,t_i,o_i\rangle$ on which they implemented Algorithm~\ref{alg:Data encryption algorithm} prior to sending WiFi data to us. In this data, each of 2000+ APs is considered as a location. We created two types of WiFi datasets: (\textit{i}) a small dataset of 26M WiFi connectivity rows collected over 44 days, and (\textit{ii}) a large dataset of 136M rows (of 14GB) collected over 202 days. For \textsc{Concealer} Algorithm~\ref{alg:Data encryption algorithm}, which produces encrypted data as shown in Table~\ref{tab:Encrypted data at the cloud}, we fixed a grid of $490 \times 16,000$ and allocated 87,000 cell-ids that resulted in two vectors $\mathit{cell\_id}[]$ and $\mathit{c\_tuple}[]$ of size 31MB. Data was encrypted using AES-256. This dataset has also skewed over the number of tuples at locations in a given time. For example, the minimum number of rows at all locations in an hour was $\approx$6,000, while the maximum number of rows at all locations in an hour was $\approx$50,000.


\medskip\noindent\textbf{Dataset 2: TPC-H dataset.} Since WiFi dataset has only three columns, to evaluate \textsc{Concealer}'s practicality in other types of data with more columns, we used 136M rows of LineItem table of TPC-H benchmark. We selected nine columns (Orderkey (OK), Partkey (PK), Suppkey(SK), Linenumber (LN), Quantity, Extendedprice, Discount, Tax, Returnflag). This dataset contains large domain values, also; \textit{e}.\textit{g}., in OK column, the domain value varies for 1 to 34M. We created
(\textit{i}) two indexes on attributes $\langle$OK, LN$\rangle$ and $\langle$OK, PK, SK, LN$\rangle$,
(\textit{ii}) two filters on concatenated values of $\langle$OK, LN$\rangle$ and $\langle$OK, PK, SK, LN$\rangle$, and
(\textit{iii}) one value that is the encryption of the concatenated values of all remaining five attributes. We used a $112,000 \times 7$ grid for index $\langle$OK, LN$\rangle$
and a $1500 \times 100 \times 10 \times 7$ grid for index $\langle$OK, PK, SK, LN$\rangle$. Each grid was allocated 87,000 cell-ids. The size of $\mathit{cell\_id}[]$ and $\mathit{c\_tuple}[]$ vectors for both grids was 54MB. Data is encrypted using Algorithm~\ref{alg:Data encryption algorithm} and AES-256. 

\medskip\noindent\textbf{Queries.} Table~\ref{tab:Sample queries} lists sample queries supported by \textsc{Concealer} on spatial time-series data. These queries as mentioned in~\S\ref{subsec:Entities and Assumptions} provide two applications: aggregate (Q1-Q3) and individualized (Q4-Q5). On TPH-C data, we executed count, sum, min/max queries.

\medskip\noindent
\textbf{Setup.} The IT department (worked as $\mathcal{DP}$) had a machine of 16GB RAM. Our side (worked as $\mathcal{SP}$) had a 16GB RAM Intel Xeon E3 machine with Intel SGX. At $\mathcal{SP}$, MySQL is used to store data, and $\approx$8,000 lines of code in C language is written query execution. 


We evaluate both versions of \textsc{Concealer} depending on the security of SGX:
(\textit{i}) one that assumes SGX to be completely secure against side-channel attacks, denoted by\textsc{Concealer}, and
(\textit{ii}) another that assumes SGX is not secure against side-channel attack (cache-line, branch shadow, page-fault attacks) and hence performs the oblivious computation in SGX (given in~\S\ref{subsec:Oblivious Trapdoor Creation and Filtering Steps}). denoted by \textsc{Concealer}+. In all our experiments, \textcolor[rgb]{1.00,0.00,0.00}{we show the overhead of preventing the side-channel attacks using red color}.



\begin{figure*}[!t]
		\B
		\begin{center}
\begin{minipage}{.99\linewidth}
				\centering
    \includegraphics[scale=0.40]{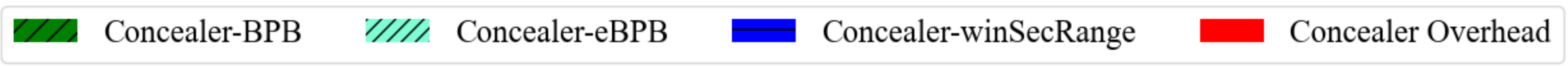}
    \BB
			\end{minipage}
			
			\begin{minipage}{.49\linewidth}
				\centering
    \includegraphics[scale=0.26]{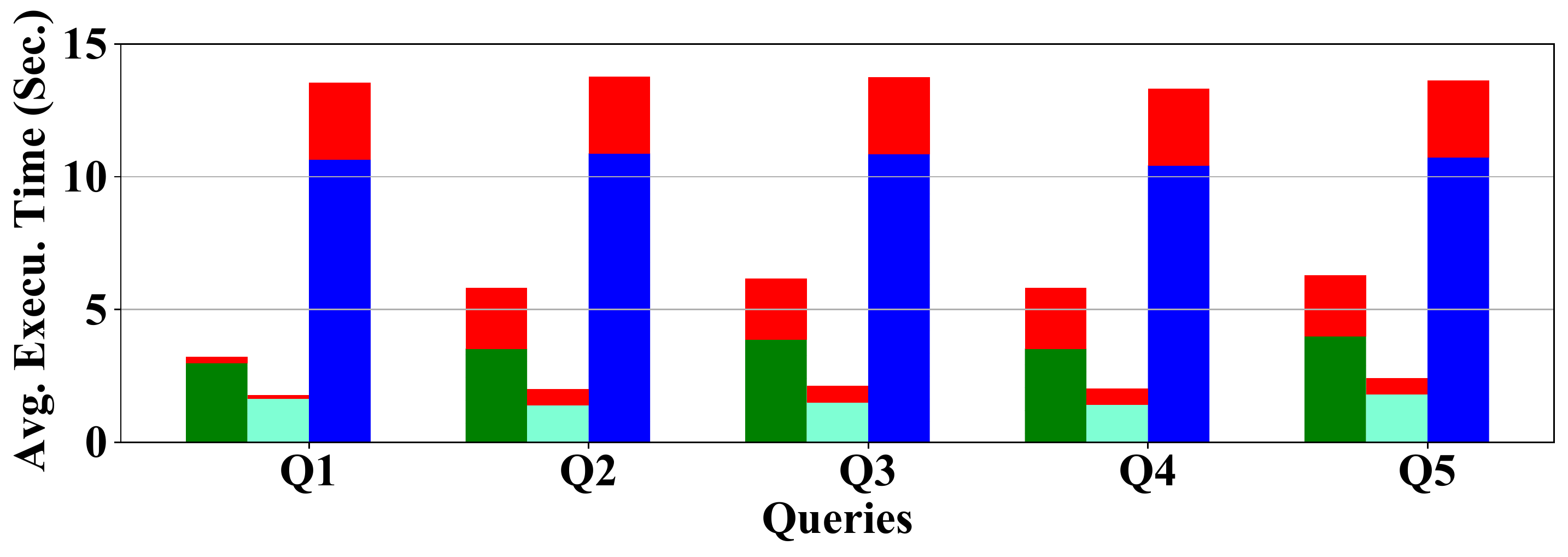}
    \BBB
    \caption{Exp 2: Range queries on 26M tuples.}
    \label{img-range-query-26}
			\end{minipage}			
			\begin{minipage}{.49\linewidth}
				\centering
    \includegraphics[scale=0.26]{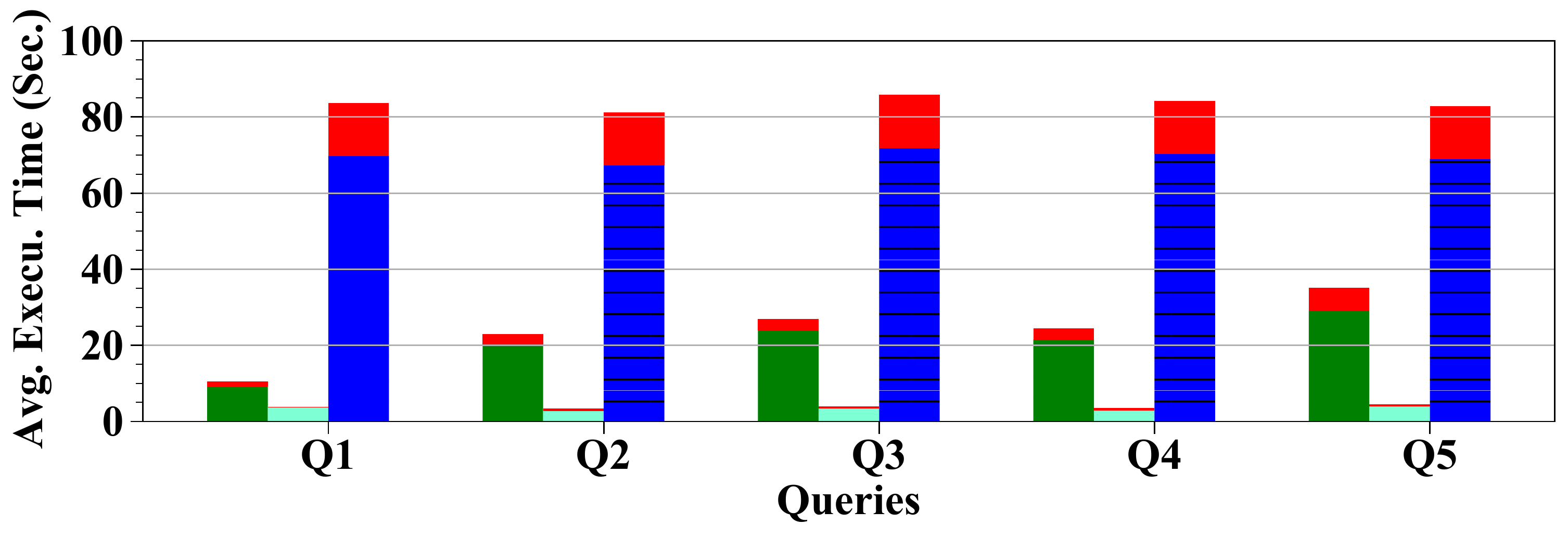}
    \BBB
    \caption{Exp 2: Range queries on 136M tuples.}
    \label{img-range-query-136}
			\end{minipage}
		\end{center}
		\BBB
	\end{figure*}

\subsection{\textsc{Concealer} Evaluation}
\label{subsec:concealer evaluation}
This section evaluates \textsc{Concealer} on different aspects such as scalability, dynamic data insertion, the impact of the range length, and the number of cell-ids.

\medskip\noindent\textbf{Exp 1: Throughput.} Since \textsc{Concealer} is designed to deal with data collected during an epoch arriving continuously over time, we measured the throughput (rows/minute) that \textsc{Concealer} can sustain to evaluate its overhead at the ingestion time. Algorithm~\ref{alg:Data encryption algorithm} can encrypt $\approx$37,185 WiFi connectivity tuple per minute. Also, it sustains our organization-level workload on the relatively weaker machine used for hosting \textsc{Concealer}.

%
%
%
%
%

\medskip\noindent\textbf{Exp 2: Scalability of \textsc{Concealer}.} To evaluate the scalability  of  \textsc{Concealer}, we compare the five queries as specified in Table~\ref{tab:Sample queries} on two WiFi datasets.


\medskip\noindent\textbf{\textit{Point query.}} For point query, we executed
a variant of Q1 when the time is fixed to be a point (instead of a range). Table~\ref{tab:point query} shows the average time taken by 5 randomly selected point queries (each executed 10 times). Note that, in \textsc{Concealer}, since the time taken by point queries is dependent upon the number of tuples allocated to the same cell-id (\textit{i}.\textit{e}., the bin size) that was 2,378 rows from small and 6,095 rows from large datasets. Table~\ref{tab:point query} shows that \textsc{Concealer} with secure SGX using BPB method took 0.23s on small and 0.90s on large datasets, while \textsc{Concealer}+ with the current non-secure SGX using BPB method took 0.37s on small and 1.38s on large datasets. \textbf{\emph{Time in \textsc{Concealer}+ increases compared to \textsc{Concealer}}}, since we need to obliviously form the queries and obliviously filter the tuples in \textsc{Concealer}+, (and that implements a data-independent sorting algorithm; see~\S\ref{subsec:Oblivious Trapdoor Creation and Filtering Steps}). Here, executing the same query on cleartext data took 0.03s on small and 0.05s on large datasets.

\begin{table}[!h]
\BB
    \centering
    \footnotesize
    \begin{tabular}{|l|l|l|l|l|}\hline
                    ~               & Small dataset (26M) & Large dataset (136M) \\\hline
        Cleartext processing    & 0.03s & 0.05s \\\hline
        \textsc{Concealer} (secure SGX) & 0.23s & 0.90s \\\hline
        \textsc{Concealer}+ (non-secure SGX) & 0.37s & 1.38s  \\\hline
    \end{tabular}
    \caption{Exp 2: Scalability of point query.}
    \label{tab:point query}
    \BBB\BBB\BBB
\end{table}


\medskip\noindent\textbf{\textit{Range queries.}} To evaluate range queries, we set the default time range for queries Q1-Q5 specified in Table~\ref{tab:Sample queries} to 20min (Exp 4 will study the impact of different range lengths). Figures~\ref{img-range-query-26} and~\ref{img-range-query-136} show
the results as an average over 5 queries (each executed 10 times). We compare BPB, eBPB (\S\ref{subsec:Enhanced Bin-Packing-Based eBPB Method}), and winSecRange (\S\ref{subsec:Optimizing Range Query}) with both \textsc{Concealer} and \textsc{Concealer}+.

Recall that BPB method answers a range query by converting it into many point queries and fetches bins corresponding to each point query; while eBPB method fetches rows corresponding to top-$\ell$ cells, which cover the given range. In \textsc{Concealer}, a cell covers $\approx$18min. Thus, for a range of 20min, BPB and eBPB methods fetch at most 3 bins and at most 3 cells, respectively, for query Q1. Thus, for answering Q1, BPB fetches $\approx$6K rows from small and $\approx$18K rows from large datasets, and eBPB fetches $\approx$1.5K rows from small and $\approx$3K rows from the large dataset.
Since eBPB retrieves few numbers of rows compared to BPB, in sSGX, eBPB performs better than BPB (see Figures~\ref{img-range-query-26} and~\ref{img-range-query-136}). \textsc{Concealer}+ again takes more time compared to \textsc{Concealer} for both eBPB and BPB, due to oblivious operations.
Note that in queries Q2-Q5, we use more locations; thus, the number of rows retrieved changes accordingly, and hence, the processing time also changes, as shown in Figures~\ref{img-range-query-26} and~\ref{img-range-query-136}.

For winSecRange, we set the range length on the time attribute to 8 hours in case of small and $\approx$1-day in case of large datasets. Thus, by fetching data for 1-day in the case of the large dataset, the enclave can execute any range query that is of a smaller time length. As expected, winSecRange took more time to execute range queries on both datasets, since it fetches and processes more rows ($\approx$70K rows from small and $\approx$400K from large datsets). While it takes more time compared to BPB and eBPB, it prevents the attack by sliding the time window (as shown in Example~\ref{subsec:Enhanced Bin-Packing-Based eBPB Method}.2), thereby, prevents revealing output-size attacks due to the sliding time window. Further, winSecRange under \textsc{Concealer}+ took more time compared to winSecRange under \textsc{Concealer}, due to oblivious operations. Recall that under \textsc{Concealer}, SGX architecture is vulnerable to side-channel attacks.

\begin{figure}[!t]
		\B
				\centering
    \includegraphics[scale=0.4]{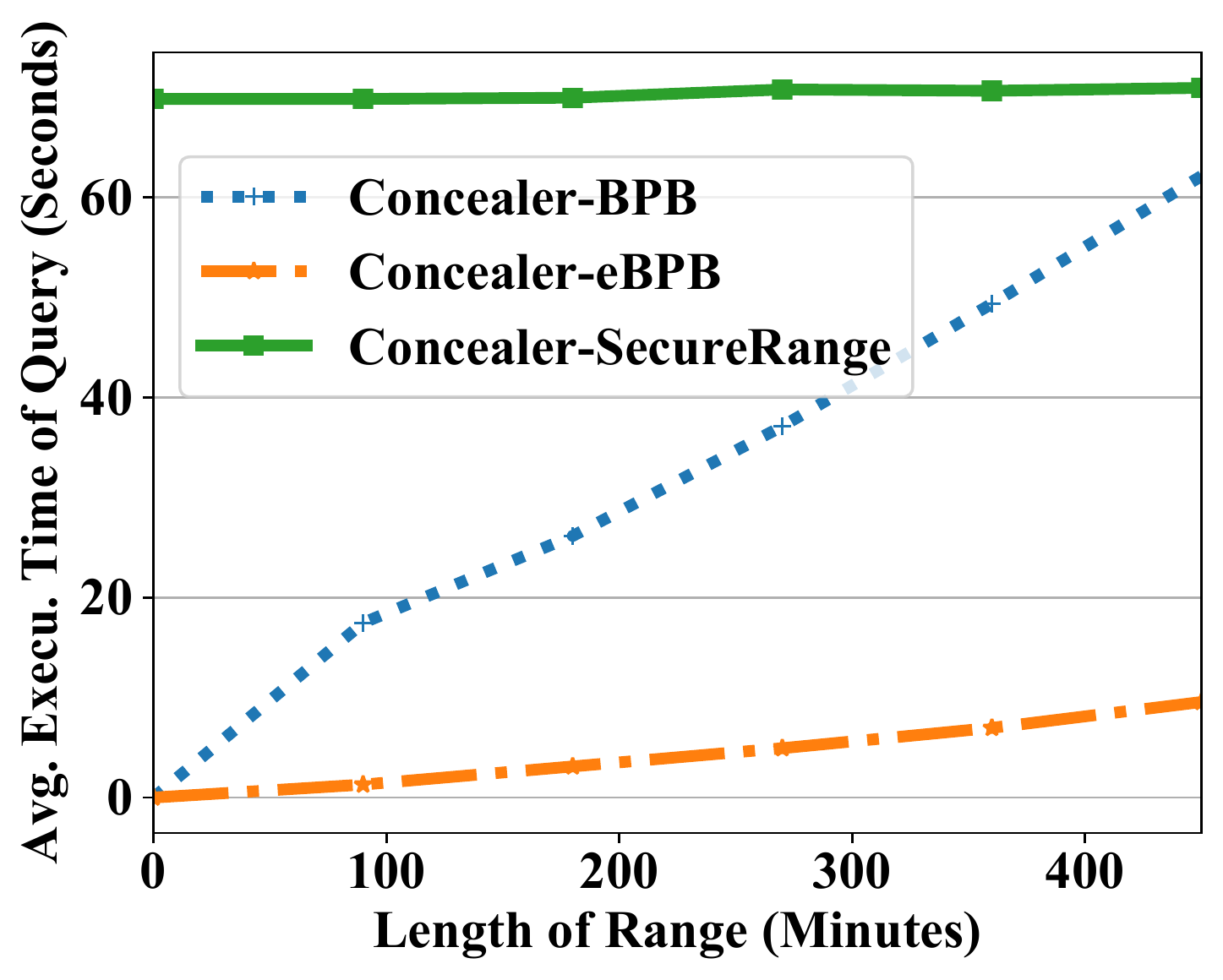}
    \BBB
    \caption{Exp 3: Range length impact.}
        \label{graph:Impact of range length}
	\end{figure}

\medskip\noindent\textbf{Exp 3. Impact of range length.} Figure~\ref{graph:Impact of range length} shows the impact of the length in a range query on \textsc{Concealer}, by executing Q1 (see Table~\ref{tab:Sample queries}) with different time lengths over the large dataset and compares three approaches BPB, eBPB, and winSecRange. In \textsc{Concealer}, a cell covers $\approx$18min. Thus, for instance, for a range of 100min, BPB and eBPB methods fetch at most 7 bins and at most 7 cells, respectively. As expected, as the length of range increases, the number of rows to be fetched from the DBMS also increases, thereby, the processing time at secure hardware increases. As mentioned in Exp 2, for the large dataset, the range length is set to $\approx$1-day for winSecRange method; hence, fetching/processing more tuples takes more time and remains almost constant for the given length of queries.


\medskip\noindent
\textbf{Exp 4.: Verification overhead.} SGX executes verification protocols by forming a hash chain on the retrieved tuples, and thus, the verification overhead is propositional to the number of retrieved tuples. In the point query, \textsc{Concealer} retrieves the minimum number of rows, \textit{i}.\textit{e}., 2,378 rows from small and 6,095 rows from large datasets, and verifying such rows took at most 0.09s and 0.16s, respectively. On the other hand, winSecRange method fetches the maximum number of rows, \textit{i}.\textit{e}., $\approx$70K rows from
small and $\approx$400K from large datsets, and verifying such rows took at most 0.8s and 3s, respectively. Note that the verification overhead is not very high.

\begin{table}[h]
    \centering
    \footnotesize
    \begin{tabular}{|l||l|l|l|l|}\hline
   &  \multicolumn{2}{|c|}{Point query} & \multicolumn{2}{|c|}{winSecRange}\\\hline
\#retrieved rows    & 2,376  & 6,095 & 70,000 & 400,000 \\\hline
Query execution time & 0.23  &  0.9 & 11 & 71 \\\hline
Verification overhead & 0.09  &  0.16 & 0.8 & 3 \\\hline
    \end{tabular}
    \caption{Exp 4: Query execution time vs verification overhead.}
    \label{tab:Query execution time vs verification overhead.}
    \BBB\BBB
\end{table}

\medskip\noindent\textbf{Exp 5. Impact of dynamic insertion.}
We also investigated how does \textsc{Concealer} support dynamic insertion of WiFi dataset. We initiated Algorithm~\ref{alg:Data encryption algorithm} for an hour of WiFi data at the peak hour, which included $\approx$50K tuples. For each insertion round, the grid size was $20 \times 1,250$ with 400 allocated cell-ids, and vectors $\mathit{cell\_id}[]$ and $\mathit{c\_tuple}[]$ of size $\approx$100KB were generated. In non-peak hours, we received at least $\approx$6K real rows. Recall that we are not interested in hiding peak vs non-peak hour data. Thus, all rows of each hour were sent using Algorithm~\ref{alg:Data encryption algorithm}. The query execution performance on dynamically inserted data depends on the number of rounds over which a query spans. For each round, we need to load the two vectors and fetch $\log |\mathit{Bin}|$ bins, as described in~\S\ref{sec:Insert Operation}. For peak hour data, we obtained 146 bins storing $\approx$400 tuples, in each, using BPB method (\S\ref{subsec:Bin-Packing-based BPB Method Query Execution}) that resulted in $\approx$3K row retrieval. On this data, it took at most $\approx$4s to execute a query, re-encrypting tuples, and writing them, for \textsc{Concealer}.

\medskip\noindent\textbf{Exp 6. Impact of bin-size.} As we create bins over real tuples and fake tuples, we study the impact of bin-size on the number of real and fake tuples included in a bin.
We vary the bin-size from 6,100 to 7,900 for answering a point query; see Figure~\ref{graph:Impact of bin-size}. As we know, first-fit-decreasing (FFD) ensures that all the bins should be half-full, except for the last one. On our dataset, FFD works well, and the bins contain most of the real tuples, as shown in Figure~\ref{graph:Impact of bin-size}. It shows that while changing the bin-size, we do not add more fake tuples to each bin.

\begin{figure}[!t]
		\BBB
		\begin{center}
			\begin{minipage}{.98\linewidth}
				\centering
    \includegraphics[scale=0.4]{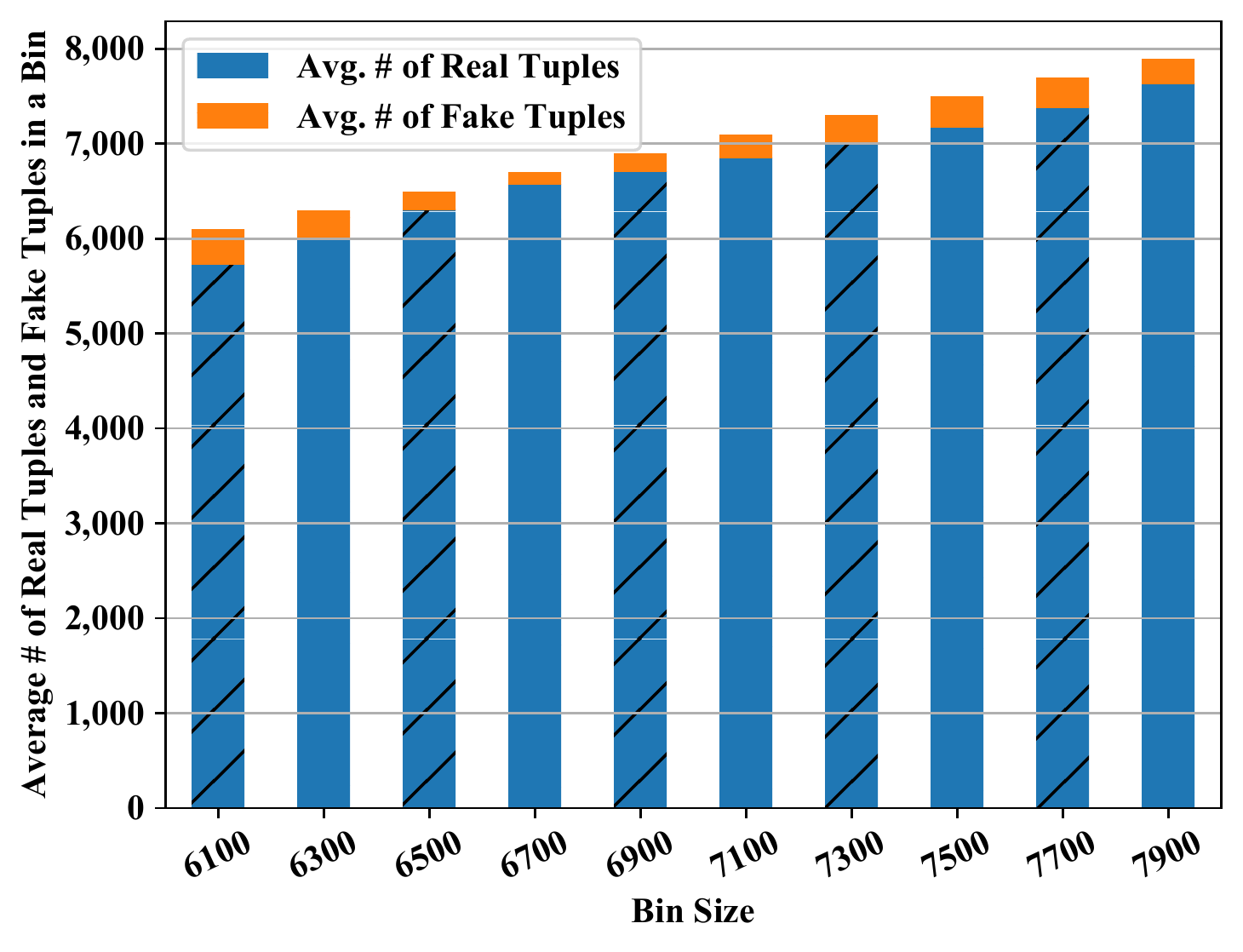}
    \BBB
    \caption{Exp 6: Impact of bin-size.}
    \label{graph:Impact of bin-size}
			\end{minipage}			
		\end{center}
	\end{figure}

\medskip
\noindent\textbf{Exp 7. Impact of the number of cells.} \textsc{Concealer} is based on the number of cell-ids allocated to the grid to execute a query. The number of cell-ids impacts the performance of the query execution significantly. To measure the impact of the number of cell-ids, we allocated a different number of cell-ids to our grid of $16,000\times 490$ and executed a point query over the large dataset. Figure~\ref{graph:impact of cell-ids} shows that when we allocate only a few cell-ids to the grid, we need to fetch a significant amount of data from the DBMS, since several cells of the grid hold the same cell-id, which in turn increase the bin-size. In contrast, as the number of cell-ids increases, we need to fetch lesser data from the DBMS, since by allocating several cells-ids, a cell-id is allocated to fewer tuples, which, in turn, decrease the bin-size.

\begin{figure}[!h]
		\BBB
		\begin{center}
			\begin{minipage}{.98\linewidth}
				\centering
 \includegraphics[scale=0.4]{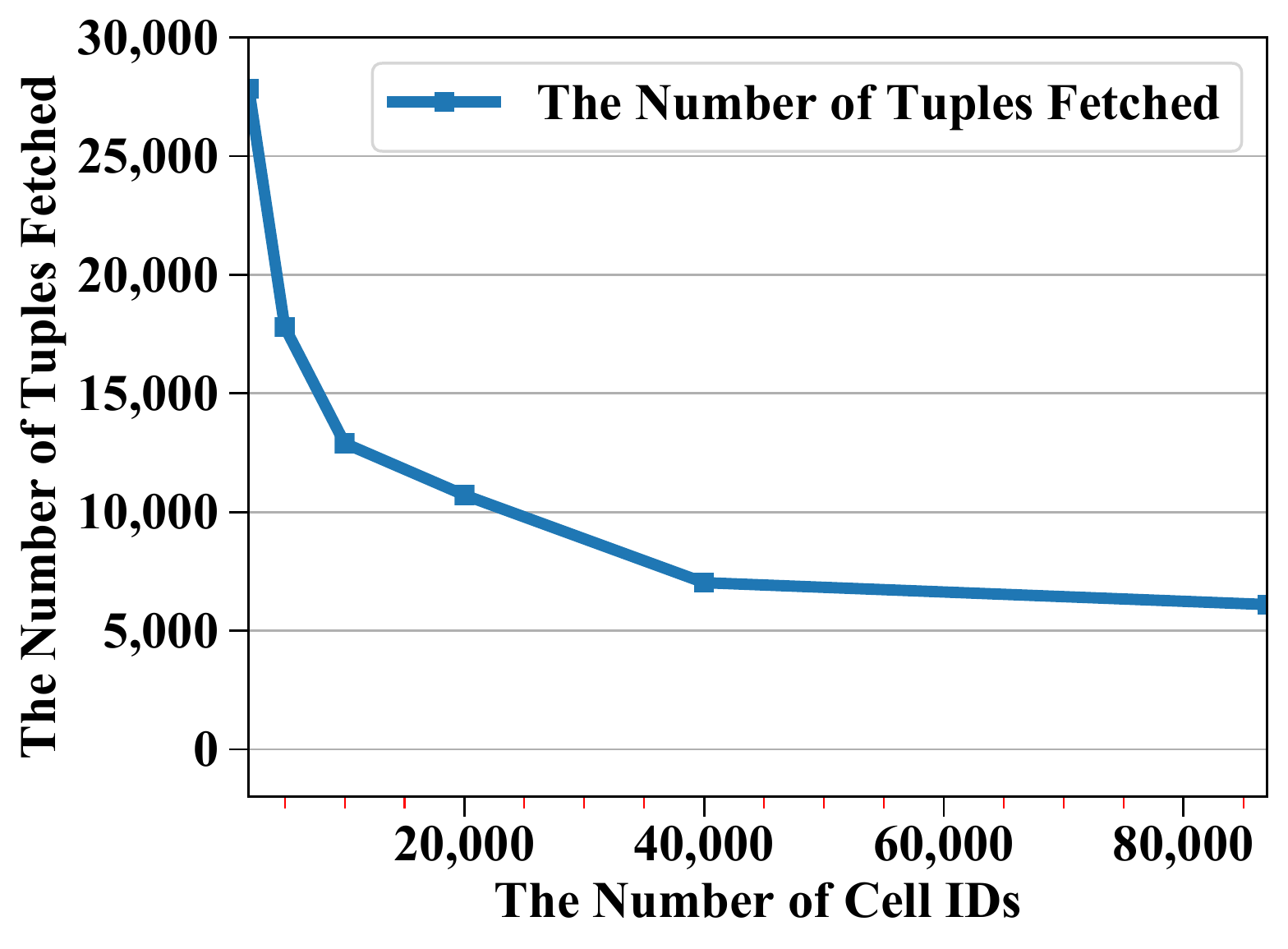}
\BBB
     \caption{Exp 7: Impact of \# cells-ids.}
     \label{graph:impact of cell-ids}
 				\end{minipage}			
		\end{center}
	\end{figure}

\begin{figure}[!h]
		\BBB\BBB
				\centering
    \includegraphics[scale=0.3]{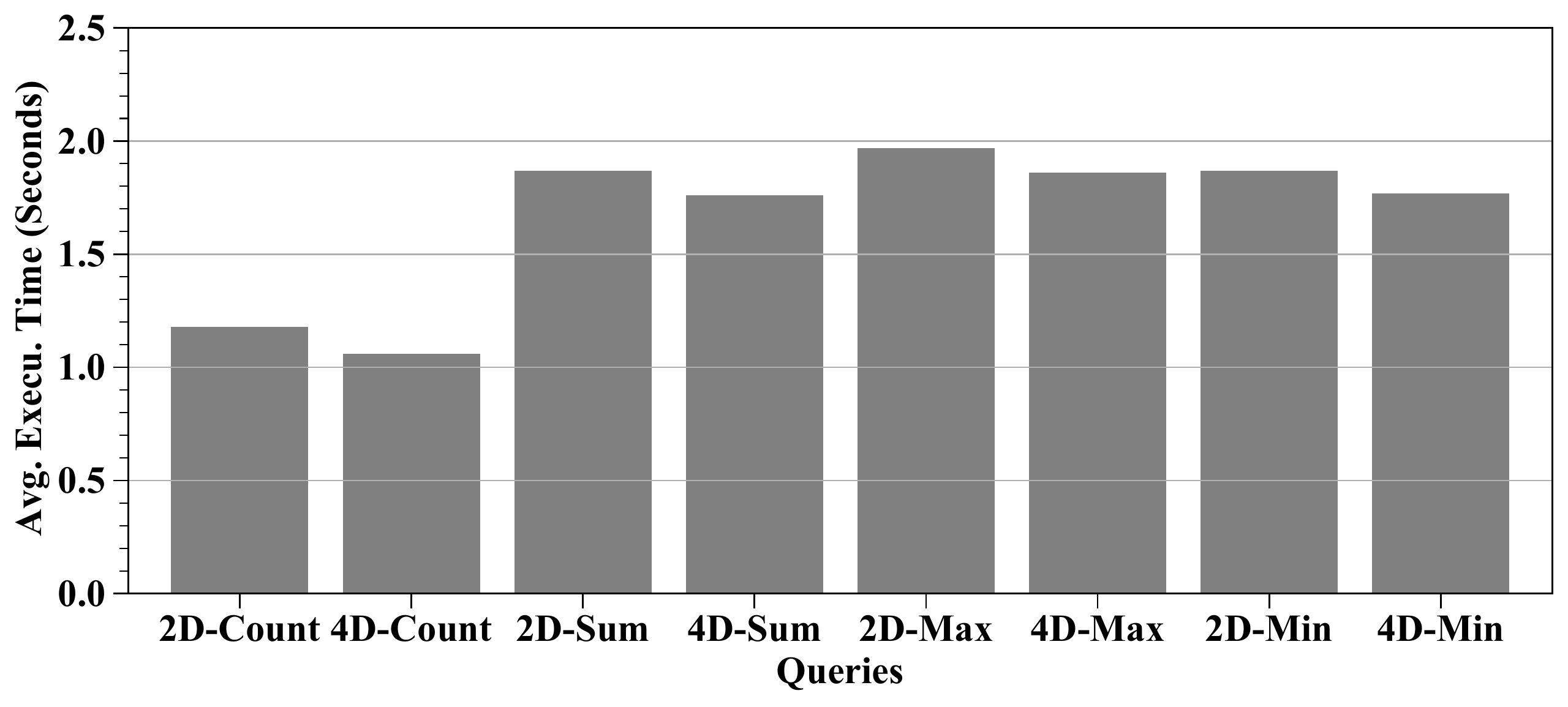}
    \BBB
    \caption{Exp 8: Query performance on TPC-H data.}
    \label{graph:tpc}
	\end{figure}
	
\medskip
\noindent\textbf{Exp 8. \textsc{Concealer} on TPC-H data.} To evaluate \textsc{Concealer}'s practicality in other types of queries, we executed two-dimensional (2D) and four-dimensional (4D) count, sum, maximum, and minimum queries on LineItem table using \textsc{Concealer}. 2D (and 4D) queries involved OK and LN (and OK, PK, SK and LN) attributes. Similar to the point query execution on WiFi dataset, 2D and 4D queries on TPC-H data require to fetch tuples allocated the same cell-id (in the same bin) according to Algorithm~\ref{alg:Query execution}. Thus, the query execution performance is dependent on the bin size, which was 400 rows for 2D grid and 6,258 for 4D grid. The query execution results are shown in Figure~\ref{graph:tpc}.

Figure~\ref{graph:tpc} shows that \textsc{Concealer} using BPB method took $\approx$1s to 2s on TPC-H dataset. Also, observe that the performance of count queries is $\approx$36\%--40\% better than the others queries, since count queries do not need to decrypt retrieved rows and it executed string matching on the filter column to produce the answer. In contrast, other queries that require exact values of the attributes decrypt retrieved rows, and hence, incur the overhead.


\subsection{Other Cryptographic Techniques}
\label{subsec:Comparison with Other Work}
Since in our setting $\mathcal{SP}$ uses secure hardware, we need to compare against a system that can support database operations using SGX. Thus, we selected an open-source SGX-based system: Opaque~\cite{DBLP:conf/nsdi/ZhengDBPGS17}.

\medskip
\noindent\textbf{\textit{Comparison between Opaque and \textsc{Concealer}.}} Opaque supports mechanisms to execute databases queries on encrypted data by first reading the entire data in the enclave, decrypting them, and then providing the answer. Note that both Opaque and \textsc{Concealer}  assume that SGX is secure against side-channel attacks, and hence, both reveal access-patterns. Thus, this is a fair comparison of the two systems, while \textsc{Concealer} avoids reading the entire dataset due to using the index and pushing down the selection predicate. Under this comparison, we execute point and range queries using Opaque and \textsc{Concealer}.

Further, note that since \textsc{Concealer}+ completely hides access-patterns inside SGX and partially hides access-patterns when fetching data in the form of fixed-size bins from the disk, we do not directly compare \textsc{Concealer}+ and Opaque due to different level of security offered by two systems.


\medskip
\noindent\textbf{Exp 9: Point queries on WiFi data.} {Opaque took more than 10min} on both WiFi datasets for executing a variant of Q1 when the time is fixed to be a point, since {\emph{Opaque requires reading the entire dataset}}. For the same query, \textsc{Concealer} took at most 0.23s on 26M and 0.9s on 136M rows.

Further, to execute the same query, \textsc{Concealer}+ took $\approx$1.4s. Thus, it shows that \textsc{Concealer}+, which hides access-patterns inside the enclave and prevents the output-size attack, is significantly better than Opaque.


\medskip
\noindent\textbf{Exp 10: Range queries on WiFi data.} {{In all range queries Q1-Q5 on WiFi data, \textsc{Concealer}'s eBPB and winSecRange algorithms take at most 4s and 71.9s over the large dataset compared to {Opaque that took at least 10min in any query}}}; see  Table~\ref{tab:Exp 8 Range queries Opaque vs SecureTimeDB}.

Further, note that to execute the same queries, \textsc{Concealer}+ takes at most 90s over the large dataset, which shows better performance of \textsc{Concealer}+ than Opaque in the case of range queries also.

\begin{table}[!h]
\BB
\footnotesize
    \centering
    \begin{tabular}{|l|l|l|l|l|l|l|l|}\hline
        System  & Q1 & Q2 & Q3 & Q4 & Q5 \\\hline
         Opaque & >10 m & >10 m &>10 m &>10 m & >10 m  \\\hline
\textsc{Concealer} eBPB & 3.6 s  & 2.8 s & 3.4 s & 3 s &  4s \\\hline
\textsc{Concealer} winSecRange & 70 s  & 67.2 s & 71.9 s & 70.2 s  & 68.9 s \\\hline
    \end{tabular}
    \caption{Exp 10: Range queries: Opaque vs \textsc{Concealer}.}
    \label{tab:Exp 8 Range queries Opaque vs SecureTimeDB}
    \BBB\BBB\BBB
\end{table}



\medskip
\noindent\textbf{Note.}
Except for Opaque, we did not experimentally compare \textsc{Concealer} against cryptographic techniques, since such techniques either offer different security levels \cite{DBLP:journals/pvldb/LiLWB14,DBLP:conf/icde/LiL17,DBLP:conf/sigmod/DemertzisPPDG16,DBLP:conf/esorics/KerschbaumT19}, or do not scale to large-sized data (\textit{e}.\textit{g}.,~\cite{DBLP:conf/ctrsa/IshaiKLO16,DBLP:journals/iacr/ArcherBLKNPSW18}) for which we have designed \textsc{Concealer}, or are not publicly available. Thus, we decide to compare \textsc{Concealer} results
with the reported result in different papers. Previous works on secure OLAP queries either support limited operations, reveal data due to DET or OPE, or scale to a smaller dataset. For example, Monomi~\cite{DBLP:journals/pvldb/TuKMZ13}, Seabed~\cite{DBLP:conf/osdi/PapadimitriouBC16},~\cite{DBLP:conf/vldb/GeZ07}, and~\cite{DBLP:conf/dawak/LopesTMCC14} reveal data due to DET or OPE. Nevertheless, Seabed supports a huge dataset ($\approx$1.75B rows). Novel SSEs, \textit{e}.\textit{g}.,~\cite{DBLP:journals/pvldb/LiLWB14,DBLP:conf/icde/LiL17,DBLP:conf/sigmod/DemertzisPPDG16,DBLP:conf/esorics/KerschbaumT19}, are very efficient, as given in their experiments; however, they provide results over 5M rows and susceptible to output-size attacks. We also checked access-pattern-hiding cryptographic work (\textit{e}.\textit{g}., DSSE~\cite{DBLP:conf/ctrsa/IshaiKLO16} and Jana~\cite{DBLP:journals/iacr/ArcherBLKNPSW18}) that are prone to output-size attacks; however, as expected, these systems are slow due to using highly secure cryptographic techniques that incur overheads and/or do not support a large-sized dataset. 
For example, an industrial MPC-based system Jana took 9 hours to insert 1M LineItem TPC-H rows, while executing a simple query took 532s. 





\section{Conclusion}
\label{sec:Conclusion}
This paper proposed \textsc{Concealer} that blends a carefully chosen encryption method with mechanisms to add fake tuples and exploits secure hardware to efficiently answer OLAP-style queries. 
We applied \textsc{Concealer} to real spatial time-series datasets, as well as, synthetic TPC-H data, and demonstrated scalability to large-sized data. Since \textsc{Concealer} allows indexing, its performance is similar to SSEs.
\textsc{Concealer} offers two key advantages over existing SSEs: first, it does not require new data structures to incorporate into databases and leverages existing index structures of modern databases. 
Second (and perhaps more importantly), \textsc{Concealer} offers a higher level of security, in addition to being IND-CKA, which existing SSEs are, by preventing leakage of data distributions via output-size.

\end{document}